\documentclass[a4paper,twocolumn,10pt,accepted=2026-11-18]{quantumarticle}
\pdfoutput=1
\input{glyphtounicode}
\usepackage[T1]{fontenc}
\usepackage[english]{babel}
\usepackage{amsmath,amssymb,amsthm,amsfonts}
\usepackage{graphicx,dcolumn,bm,dsfont}
\usepackage{enumitem,multirow,array,float}
\usepackage{csquotes}
\MakeOuterQuote{"}
\usepackage{braket,xfrac}
\usepackage{tabularx,tikz,quantikz,booktabs}
\usepackage[numbers,sort&compress]{natbib}
\usepackage{hyperref}
\newtheorem{theorem}{Theorem}
\newtheorem{definition}{Definition}
\newtheorem{lemma}{Lemma}

\newtheorem{corollary}{Corollary}

\begin{document}
\normalfont

\title{GHZ-Preserving Gates and Optimized Distillation Circuits}
\author{Mingyuan Wang}
\email{mingyuanwang@umass.edu}
\affiliation{Department of Physics, University of Massachusetts Amherst}

\author{Guus Avis}
\affiliation{College of Information and Computer Sciences, University of Massachusetts Amherst}

\author{Kenneth Goodenough}
\affiliation{College of Information and Computer Sciences, University of Massachusetts Amherst}

\author{Stefan Krastanov}
\email{skrastanov@umass.edu}
\affiliation{Department of Physics, University of Massachusetts Amherst}
\affiliation{College of Information and Computer Sciences, University of Massachusetts Amherst}

\maketitle

\begin{abstract} 
Greenberger-Horne-Zeilinger (GHZ) states play a central role in quantum computing and communication protocols, as a typical multipartite entanglement resource. This work introduces an efficient enumeration and simulation method for circuits that preserve and distill noisy GHZ states, significantly reducing the simulation complexity of a gate on $n$ qubits, from exponential $O(2^n)$ for standard state-vector methods or $O(n)$ for Clifford circuits, to a constant $O(1)$ for the method presented here. This method has profound implications for the design of quantum networks, where preservation and purification of entanglement with minimal resource overhead is critical. In particular, we demonstrate the use of the new method in an optimization procedure enabled by the fast simulation, that discovers GHZ distillation circuits far outperforming the state of the art. Fine-tuning to arbitrary noise models is possible as well. We also show that the method naturally extends to graph states that are local Clifford equivalent to GHZ states.

\end{abstract}

\section{Introduction}
In quantum information theory, multipartite entanglement is a key resource for a wide range of quantum protocols, from communication and cryptography to distributed quantum computing. Among the most studied forms of multipartite entanglement are Greenberger-Horne-Zeilinger (GHZ) states, which can be represented for $n$ qubits as
\begin{equation}
|\text{GHZ}\rangle = \frac{1}{\sqrt{2}} (|0\rangle^{\otimes n} + |1\rangle^{\otimes n}).
\end{equation}
GHZ states are particularly valuable due to their characteristic non-classical correlations across multiple qubits, which play a central role in fundamental quantum tasks such as quantum secret sharing~\cite{Hillery1999}, quantum conference key agreement~\cite{Murta2020,Grasselli2022,Hahn2020}, error correction~\cite{bell2014}, and distributed quantum computing~\cite{Yimsiriwattana2004,Zhong2021}. However, their practical use is limited by the fragility of their entanglement under noise and imperfect local gates. This limitation motivates the use of entanglement distillation protocols, which aim to recover high-fidelity entangled states from multiple noisy copies via local operations and classical communication (LOCC)~\cite{bennett1996,bennett1996mixed,dur2007,Acin2000}. 

The computational complexity of simulating entanglement distillation and error correction can be a bottleneck to large-scale circuit optimization. The exponential cost $O(2^n)$ of state-vector simulation of $n$ qubits can easily be avoided thanks to the more efficient Clifford-circuits formalism and its $O(n)$ scaling for gate application~\cite{aaronson2004,Gottesman1998}. However, even $O(n)$ can become prohibitive if it is employed in the inner loop of other algorithms, like circuit optimizers over large systems, or when real-time processing is required. This highlights the need for more efficient methods to simulate GHZ states.

Another major obstacle lies in the vastness of the Clifford-circuit design space itself. Even though Clifford operations admit efficient simulation, the number of distinct Clifford circuits grows superexponentially with the number of qubits and gate layers, making exhaustive or even heuristic search intractable. This severely limits our ability to discover optimal or near-optimal GHZ-distillation protocols, especially when combined with noisy inputs and real-world hardware limitations.

To address both challenges, we introduce a new framework based on the classification of GHZ-preserving operations—those local Clifford gates that map GHZ-basis states to other GHZ-basis states~(terms defined in the following section). This restriction dramatically reduces the search space, as it excludes circuits that cannot possibly contribute to valid distillation. We then analyze the algebraic structure of the GHZ-preserving gate set and show that it forms a finite group that factorizes into a few small easy-to-enumerate subgroups. This group structure enables exact enumeration of allowed circuits and allows us to simulate their action in constant $O(1)$ time by tracking permutations over a discrete GHZ basis. This advance enables efficient simulations even for large systems, making it feasible to model large networks and to employ computationally expensive circuit-optimization techniques. Our framework generalizes similar techniques used for Bell pair distillation~\cite{addala2023,krastanov2019,dehaene2003}.

As a showcase of these algorithmic improvements, we use our new efficient simulation method together with {a genetic algorithm} to optimize GHZ-distillation circuits, while modeling the effects of network and gate noise. The efficient simulation enabled by our framework allows us to explore a much larger number of potential circuit designs, leading to more comprehensive optimization, and the circuits we produce significantly outperform anything else available in the literature. This is particularly relevant for quantum networks, where high-fidelity entanglement is essential for tasks such as distributed quantum computation and quantum conference key agreement.

Furthermore, we extend our method to other quantum states, such as star-shaped and complete graph states, which are local Clifford equivalent to GHZ states.

The remainder of this paper is organized as follows: In Section 2, we define the GHZ-preserving group of gates and present a group decomposition that enables efficient enumeration and gate simulation. We assume familiarity with the stabilizer formalism in this paper. Section 3 reviews current GHZ-distillation techniques, discusses noise modeling, and presents our superior distillation circuits, as well as the optimization techniques that generated them. In Section 4, we apply our decomposition to star-shaped and complete graph states, and we conclude in Section 5 with a discussion of implications to networking and other domains and future research directions. Detailed proofs of our group-decomposition theorems are provided in the Appendix~\ref{detailed_proofs}.

\section{Classification and Fast Simulation of GHZ-Preserving Clifford Operations}
\label{sec:ghzpreserving}

In this section, we formally define what it means for a unitary operation to \emph{preserve} GHZ states and categorize all such operations.
Informally, the set of "GHZ-preserving" operations is the set of gates that enable the distillation of such states (assuming noise is of the Pauli type, which is typical). We introduce both a fast simulation method for the application of such gates, and a compact enumeration for these gates. In later sections we will use these two capabilities to run circuit-optimization algorithms that find GHZ-distillation circuits much better than the state of the art.

At first, we focus on \emph{two copies} of an $n$-qubit GHZ state, each distributed among $n$ nodes. We will prove that any GHZ-preserving transformation can be built out of Pauli operations, \emph{homogeneous} unitaries consisting of the same gate applied to all nodes, (the \textit{H group}), and \emph{bilocal} unitaries consisting of the same gate applied to two nodes (the \textit{B group}). This decomposition greatly simplifies the analysis and simulation of GHZ-based protocols by reducing the computational complexity and enumerating the set of useful gates.
{These operations also generate all GHZ-preserving local Clifford operations on more than two GHZ states, as outlined in Appendix~\ref{app:multiqubit}.}

\subsection{Decomposition of the GHZ-preserving Group}
\label{subsec:ghzpreservingdef}

Consider two identical $n$-qubit GHZ states, each shared among $n$ nodes. Concretely, each node $i$ holds two qubits in total, one from each GHZ state.


\begin{definition}[GHZ basis]
\label{def:GHZbasis}
The $n$-qubit GHZ state is
\begin{equation}
\lvert \mathrm{GHZ}\rangle 
=\;
\frac{1}{\sqrt{2}}\Bigl(\lvert 0\rangle^{\otimes n} + \lvert 1\rangle^{\otimes n}\Bigr). 
\end{equation}
Any $n$-qubit state obtained from $\lvert \mathrm{GHZ}\rangle$ by local Pauli operations (and hence sharing the same stabilizer generating set, up to differences in sign) is said to be in the \emph{GHZ basis}. Concretely, this basis of $2^n$ states can be written as
\begin{equation}
\left\{\,
\Bigl(~\!\!\bigotimes_{i=1}^n P_i\Bigr)
\,\lvert \mathrm{GHZ}\rangle
\;\bigm|\;
P_i \in \{I,X,Y,Z\}
\right\}.    
\end{equation}

\end{definition}

For example, in the three-qubit case, the standard $\lvert \mathrm{GHZ}\rangle$ is stabilized by the generating operators $XXX$, $ZZI$, and $IZZ$. Any other state stabilized by these operators (up to a sign change) also lies in what we call the "GHZ basis" (see Table~\ref{tab:ghz-basis}).

\begin{table}[htbp]
\centering
\small\setlength{\tabcolsep}{2pt}
\begin{tabular}{@{}ccccc@{}}
\toprule
GHZ basis & Computational basis & $X_1 X_2 X_3$ & $Z_1 Z_2$ & $Z_2 Z_3$ \\
\midrule
$|\phi^{+++}\rangle$ & $(|000\rangle + |111\rangle)/\sqrt{2}$ & $+1$ & $+1$ & $+1$ \\
$|\phi^{++-}\rangle$ & $(|001\rangle + |110\rangle)/\sqrt{2}$ & $+1$ & $+1$ & $-1$ \\
$|\phi^{+-+}\rangle$ & $(|011\rangle + |100\rangle)/\sqrt{2}$ & $+1$ & $-1$ & $+1$ \\
$|\phi^{+--}\rangle$ & $(|010\rangle + |101\rangle)/\sqrt{2}$ & $+1$ & $-1$ & $-1$ \\
$|\phi^{-++}\rangle$ & $(|000\rangle - |111\rangle)/\sqrt{2}$ & $-1$ & $+1$ & $+1$ \\
$|\phi^{-+-}\rangle$ & $(|001\rangle - |110\rangle)/\sqrt{2}$ & $-1$ & $+1$ & $-1$ \\
$|\phi^{--+}\rangle$ & $(|011\rangle - |100\rangle)/\sqrt{2}$ & $-1$ & $-1$ & $+1$ \\
$|\phi^{---}\rangle$ & $(|010\rangle - |101\rangle)/\sqrt{2}$ & $-1$ & $-1$ & $-1$ \\
\bottomrule
\end{tabular}
\caption{The three-qubit GHZ-basis states and the corresponding phase of the standard GHZ stabilizer generators.}
\label{tab:ghz-basis}
\end{table}

\begin{definition}[GHZ preserving]
\label{def:ghzpreserving}
A unitary operator $U$
is called \emph{GHZ preserving} if, whenever it is applied to a product of GHZ-basis states, the resulting output is again a product of GHZ-basis states. E.g., for the case of a system of two GHZ states, for any $\lvert \mathrm{GHZ}_i\rangle$ and $\lvert \mathrm{GHZ}_j\rangle$ in the GHZ basis we will have
\begin{equation}
U \Bigl(
\lvert \mathrm{GHZ}_i\rangle \otimes \lvert \mathrm{GHZ}_j\rangle
\Bigr) = 
\lvert \mathrm{GHZ}_k\rangle \otimes \lvert \mathrm{GHZ}_l\rangle,    
\end{equation}
where $\lvert \mathrm{GHZ}_k\rangle$ and $\lvert \mathrm{GHZ}_l\rangle$ are also in the GHZ basis.
\end{definition}

Equivalently, a GHZ-preserving unitary on 2 $n$-qubit GHZ-basis states can be viewed as inducing a permutation on the $2^n\times 2^n$ basis formed by the tensor products of GHZ-basis states from each copy. This realization is important for the more efficient algorithm we introduce for simulating the action of a gate. These GHZ-preserving gates are the ones out of which entanglement-distillation circuits can be built, as they allow us to "move" Pauli errors between qubits such that they can be detected.

We focus on GHZ-preserving operations because typical entanglement-distillation protocols are designed to iteratively transform multiple noisy GHZ copies into a higher-fidelity GHZ state, and such procedures inherently require that each step preserves the overall GHZ structure to be meaningful.

Now we can introduce our main result (stated more formally in Eq.~\eqref{eq:main_result}). Every GHZ-preserving operation can be represented as a product of Pauli operations and operations belonging to the following groups:
\begin{enumerate}
    \item The \textbf{B group}, consisting of $2n$-qubit unitaries created by \emph{bilocally} applying two-qubit gates—i.e., two-qubit gates act only on two of the $n$ nodes.
    \item The \textbf{H group}, consisting of $2n$-qubit unitaries created by \emph{homogeneously} applying two-qubit gates—i.e., the same two-qubit gate applied at each of the $n$ nodes between its two local qubits.
\end{enumerate}

In the following sections, we describe each subgroup in detail, stating many of their properties, followed by a discussion of how these subgroups combine to form generic GHZ-preserving gates. For more detailed formal descriptions of these subgroups and the proof that these subgroups together with the Pauli operations span the entirety of the GHZ-preserving set of gates, consult Appendix~\ref{detailed_proofs}.

\begin{table*}[!tp]
\centering
\begin{tabular}{|c|c|c|c|}
\hline
XI $\to$ + XI & XI $\to$ + XI & XI $\to$ + YI & XI $\to$ + YI \\
IX $\to$ + IX & IX $\to$ + IY & IX $\to$ + IX & IX $\to$ + IY \\
ZI $\to$ + ZI & ZI $\to$ + ZI & ZI $\to$ + ZI & ZI $\to$ + ZI\\
IZ $\to$ + IZ & IZ $\to$ + IZ & IZ $\to$ + IZ & IZ $\to$ + IZ\\
\hline
\makebox[3cm][c]{$Identity$} & \makebox[3cm][c]{$I \otimes S$} & \makebox[3cm][c]{$S \otimes I$} & \makebox[3cm][c]{$S \otimes S$}

\end{tabular}
\begin{tabular}{|c|c|c|c|}
\hline
XI $\to$ + XZ & XI $\to$ + XZ & XI $\to$ + YZ & XI $\to$ + YZ \\
IX $\to$ + ZX & IX $\to$ + ZY & IX $\to$ + ZX & IX $\to$ + ZY \\
ZI $\to$ + ZI & ZI $\to$ + ZI & ZI $\to$ + ZI & ZI $\to$ + ZI \\
IZ $\to$ + IZ & IZ $\to$ + IZ & IZ $\to$ + IZ & IZ $\to$ + IZ \\
\hline
\makebox[3cm][c]{$CZ$} & \makebox[3cm][c]{$CZ \cdot I \otimes S$} & \makebox[3cm][c]{$CZ \cdot S \otimes I$} & \makebox[3cm][c]{$CZ \cdot S \otimes S$} \\
\hline
\end{tabular}
\caption{Two-qubit building-block gates, when applied "bilocally" resulting in the B group.}
\label{tab:Fgates}
\end{table*}

\subsection{Bilocal Gates (B Group)}
\label{sec:Bgroup}

We consider first the class of gates that act \emph{bilocally} on exactly two distinct nodes. 
Concretely, they are generated by the controlled-$Z$ (CZ) gate and the single-qubit phase ($S$) gate on each of the two qubits. 
These eight gates form an Abelian group isomorphic to $\mathbb{Z}_2 \otimes \mathbb{Z}_2 \otimes \mathbb{Z}_2$, reflecting three independent binary choices:
\begin{itemize}
    \item Apply or do not apply CZ;
    \item Apply or do not apply $S$ on the first qubit;
    \item Apply or do not apply $S$ on the second qubit.
\end{itemize}
See Table~\ref{tab:Fgates} for an explicit list.

\subsection{Homogeneous Gates (H Group)}

We next examine a second class of GHZ-preserving operations, namely those that act \emph{homogeneously} across all $n$ nodes. 
In other words, each node applies the same two-qubit gate to its local pair of qubits (one qubit from each copy of the GHZ). 
A detailed enumeration in Appendix~\ref{Appendix:H} shows exactly six possibilities, all generated by the CNOT gate (in both orientations, i.e.,\ \(\mathit{CNOT}_{12}\) and \(\mathit{CNOT}_{21}\)). 
The resulting group is non-Abelian, of order 6, and isomorphic to the dihedral group \(\mathbb{D}_3\) (the symmetry group of an equilateral triangle). 
Concretely, the set of these six homogeneous gates includes the identity, SWAP, and four conditional Pauli gates. See Table~\ref{tab:Hgates} for an explicit list.

\begin{table*}[!tp]
\centering
\begin{tabular}{|c|c|c|c|c|c|}
\hline
XI $\to$ + XI & XI $\to$ + IX & XI $\to$ + XX & XI $\to$ + IX & XI $\to$ + XX & XI $\to$ + XI \\
IX $\to$ + IX & IX $\to$ + XI & IX $\to$ + IX & IX $\to$ + XX & IX $\to$ + XI & IX $\to$ + XX \\
ZI $\to$ + ZI & ZI $\to$ + IZ & ZI $\to$ + ZI & ZI $\to$ + ZZ & ZI $\to$ + IZ & ZI $\to$ + ZZ \\
IZ $\to$ + IZ & IZ $\to$ + ZI & IZ $\to$ + ZZ & IZ $\to$ + ZI & IZ $\to$ + ZZ & IZ $\to$ + IZ \\
\hline
\makebox[2cm][c]{$Identity$} & \makebox[2cm][c]{$SWAP$} & \makebox[2cm][c]{$CNOT_{12}$} & \makebox[2cm][c]{$DCX_{21}$} & \makebox[2cm][c]{$DCX_{12}$} & \makebox[2cm][c]{$CNOT_{21}$} \\
\hline
\end{tabular}
\caption{Two-qubit building-block gates, when applied "homogeneously" resulting in the H group.}
\label{tab:Hgates}
\end{table*}

\subsection{Complete Decomposition of GHZ-Preserving Operations}

Based on the group structure described in the preceding paragraphs, any GHZ-preserving gate acting on two $n$-qubit GHZ-basis states can be decomposed as follows (see Fig.~\ref{circuits} for a depiction and the Appendix for a detailed proof):

\begin{itemize}
    \item one gate from the H group 
    \item $n-1$ gates from the B group, each acting on a different pair of nodes
    \item Pauli gates applied on \(n-1\) nodes, for each of the two GHZ-basis states (i.e., \(2n - 2\) single-qubit Pauli gates in total).
\end{itemize}

More precisely, every element $g$ in the GHZ-preserving group can be written in the form
\begin{equation} \label{eq:main_result}
    g = \prod_{j=1}^{2n-2} p_j \prod_{i=1}^{n-1} b_{i, i+1}  h ,
\end{equation}
where $h$ is a $2n$-qubit unitary from the H group, each $b_{i, i+1}$ is a $2n$-qubit unitary from the B group, and each $p_j$ is a single-qubit Pauli operator acting on a different qubit.

\begin{figure}
        \includegraphics[width=0.45\textwidth]{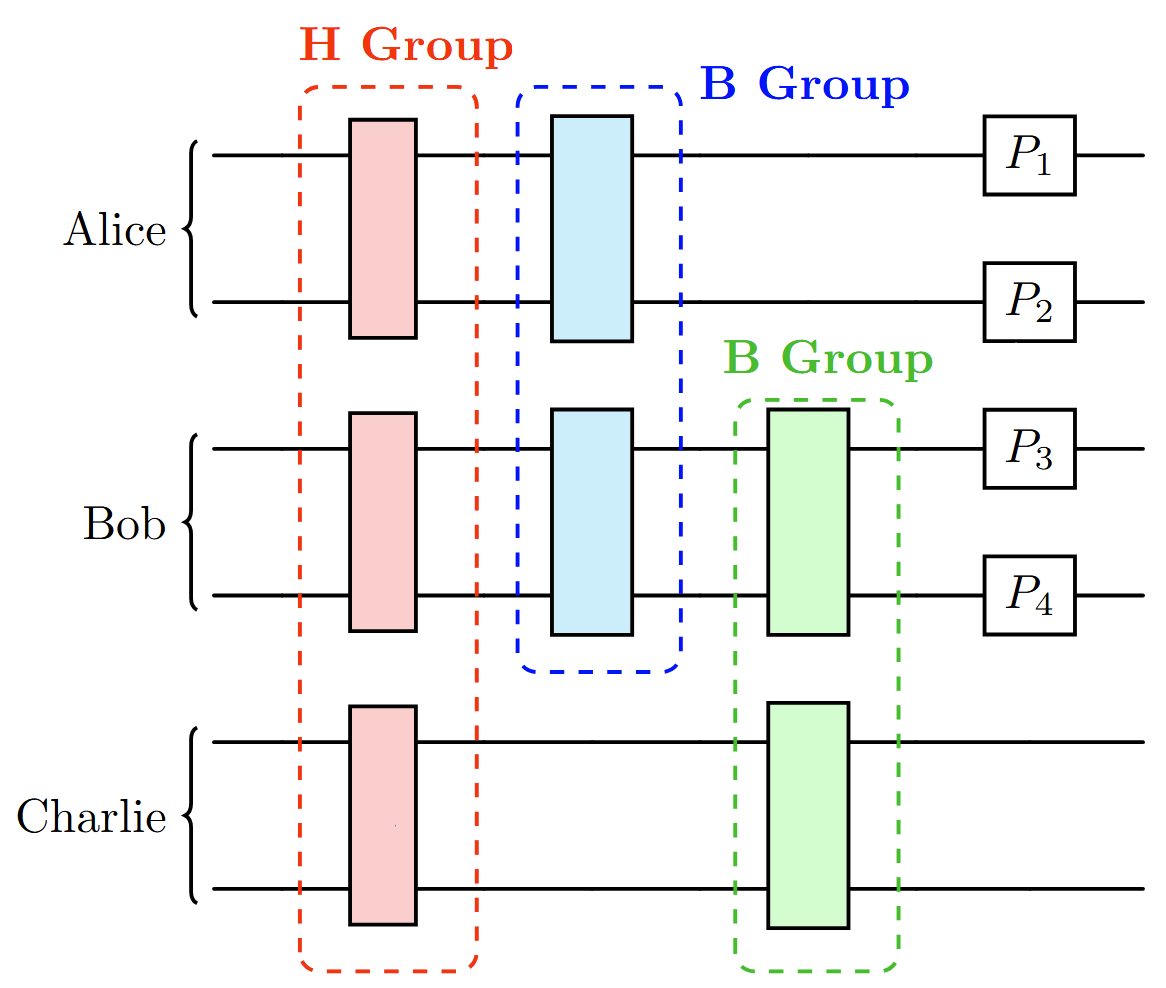}
        \caption{\textbf{General GHZ-preserving circuits.} Example with two 3-qubit GHZ states, every node holding one qubit from each state (lines 1, 3 and 5 correspond to the first GHZ state, lines 2, 4 and 6 correspond to the second GHZ state). The red blocks (left) together make up a unitary from the H group, where all nodes apply the same gate chosen from $\{I, SWAP, CNOT_{12}, CNOT_{21},DCX_{12},DCX_{21}\}$. The blue (second from left) and green (third from left) blocks make up different unitaries from the B group, where $n-1$ pairs of nodes each apply the same gate chosen from $\left\{ (g\cdot(f_1 \otimes f_2))\, |\, g \in \{I \otimes I, CZ\}, (f_1, f_2) \in \{I, S\}^2\right\} $ (see Tables~\ref{tab:Fgates} and \ref{tab:Hgates}). Lastly, we need $2n-2$ Pauli operations, $P_1,P_2,P_3,P_4$.}
        \label{circuits}
\end{figure}
    
For two $n$-qubit GHZ-basis states, the total number of possible "phaseless" GHZ-preserving unitaries (obtained by ignoring the Pauli gates in the above decomposition) is $6\times 8^{n-1}$. Even when multiple GHZ states are involved, each gate still acts on only two of them at a time, so this decomposition fully captures all possible circuit components. Besides the formal proof in the Appendix, we also numerically verified the validity of this decomposition by exhaustively enumerating all possible combinations of Clifford operations up to $n=6$ and counting the ones that are GHZ preserving (for similar numerical enumerations consult also~\cite{jansen2022}). 
Of note is not only that this decomposition gives a complete and non-redundant representation of GHZ-preserving unitaries, but also that each gate can be described by a fixed-length sequence of operations, enabling highly efficient circuit simulations and making it much easier to search over all such unitaries in a structured way.

\subsection{Binary GHZ Phase Labels and Fast Gate Simulation}

Given the decomposition of GHZ-preserving operations into the H and B groups, we now introduce a compact and computationally efficient representation of GHZ-basis states.

In systems composed of multiple GHZ-basis states, it is typical to work with block-diagonal stabilizer tableaux representing the states. All GHZ-preserving operations preserve this block-diagonal structure. Their action alters only the phases of stabilizer generators — i.e.,\ flipping signs in the phase column — while keeping the stabilizer generators fixed. This is equivalent to permuting the GHZ basis and never introducing off-diagonal terms or a superposition of different GHZ-basis states.

This motivates a compact representation: instead of tracking the full tableau or the underlying density matrix, we encode the system state as a binary string recording the sign (phase) of each stabilizer generator, where we adopt the convention that a $+$ sign is represented by 0 and a $-$ sign by 1. For a configuration of $m$ GHZ-basis states, each on $n$ qubits, the system can be completely described by an $n\times m$-bit string~(see Fig.~\ref{eq:ghz-tensor-stabilizer}).

Moreover, we note that each GHZ-preserving operation can be described as a permutation on the Cartesian product of \(m\) GHZ bases.
This is because the operations are GHZ preserving, i.e., each product of GHZ-basis states is mapped to another product of GHZ-basis states, and because the operations are unitary, i.e., the map must be invertible.

This implies that, in the bit-string representation introduced above, each GHZ-preserving operation is nothing more than a permutation of the $2^{nm}$ different bit strings.
This can be efficiently implemented by representing each bit string as an integer from 1 to $2^{nm}$; GHZ-preserving operations are then just permutations of those numbers (i.e., they are elements of the symmetric group $S(2^{nm})$).
Moreover, as each element of the H and B groups acts only on two GHZ states at a time, they are effectively only permutations between bit strings of length $2n$ and hence can be modeled more efficiently as permutations of the numbers 1 to $2^{2n}$.
Representing gates as such permutations results in a simulation complexity of \(O(1)\), i.e., a lookup in an ordered table encoding the permutation associated with the given gate. 

Although the permutation table for an $n$-qubit GHZ state has size $2^{2n}$, this does not cause a large overhead in practice: once $n$ is fixed, the table is generated only once and can be reused permanently. In terms of memory, the table stores only integer indices, which is far more compact than storing full stabilizer tableaux or density matrices; in fact the memory footprint is orders of magnitude smaller, and memory itself is cheap compared to runtime overhead.

\begin{figure*}[!tp] 
\begin{equation*}
|\phi^{+++}\rangle \otimes |\phi^{+--}\rangle \otimes |\phi^{-+-}\rangle \;\sim\;
\left[
\begin{array}{lll}
+ & XXX & \\
+ & Z\;Z\;I & \\
+ & I\;Z\;Z & \\
\\
+ & \phantom{XXX} & XXX \\
- & \phantom{XXX} & Z\;Z\;I \\
- & \phantom{XXX} & I\;Z\;Z \\
\\
- & \phantom{XXX} & \phantom{XXX} \quad XXX \\
+ & \phantom{XXX} & \phantom{XXX} \quad Z\;Z\;I \\
- & \phantom{XXX} & \phantom{XXX} \quad I\;Z\;Z \\
\end{array}
\right]
\quad \longrightarrow \quad 000\;011\;101
\end{equation*}
\caption{Example of ket notation, stabilizer tableau representation, and bitstring representation for the same state.}
\label{eq:ghz-tensor-stabilizer}
\end{figure*}

{The H- and B-group actions on these labels are affine binary transformations that can be implemented using swaps and XORs; explicit rules are given in Appendix~\ref{app:binary_action}. Appendix~\ref{app:complexity} separates the cost of dense-table lookup from its preprocessing and memory requirements and gives the scaling for unpacked and bit-packed evaluation.}

\section{GHZ Entanglement distillation and Circuit Optimization}
In this section, we provide an overview of existing GHZ-distillation protocols, discuss how noise can be modeled in quantum systems, and present our approach to optimizing distillation circuits under realistic error conditions. Our method leverages both a significant reduction of the search space—based on the decomposition framework developed in previous sections—and a drastic acceleration of circuit simulation, enabling efficient evaluation of candidate circuits under noise within the cost function of the optimizer. While in this particular section we discuss a genetic-algorithm optimizer, the choice of optimizer is not of importance in this work. Rather, the main contribution lies in the combination of search-space reduction and accelerated simulation, which together make the optimization of large circuits practically feasible.

\subsection{Overview of Entanglement Distillation for GHZ States}
Because GHZ states involve multiple qubits in a highly entangled configuration, they are especially susceptible to noise. Even small gate imperfections or external interference can degrade entanglement quality significantly, making distillation essential. 

Early entanglement-distillation protocols such as BBPSSW~\cite{bennett1996,deutsch1996} and double selection~\cite{fujii2009} rely on multiple CNOT operations and can be adapted to GHZ states via pairwise distillation or entanglement pumping~(Appendix~\ref{entanglement pumping}). Decisions on what errors to attempt to detect and how to do that are not trivial. In particular, $Z$ errors on different qubits have indistinguishable effects: they flip the relative phase. In contrast, $X$ errors lead to different states depending on which qubit it affected. This asymmetry makes it more difficult to design protocols that can simultaneously suppress both error types in GHZ systems~(Appendix~\ref{asymmetry}).

Some alternative approaches have been proposed for GHZ distillation~\cite{debond2020,glancy2006,rengaswamy2022,nickerson2013}, including protocols based on stabilizer codes or hybrid constructions with Bell pairs. While optimal under certain conditions, these methods often assume perfect gates and measurements, high-fidelity raw GHZ states/Bell pairs, or are optimal only in the asymptotic regime of infinitely large circuits, which are unrealistic in practical network settings. Moreover, the required overhead, such as a large numbers of ancilla qubits, scales poorly with system size.

In contrast, our work addresses a more realistic and challenging regime. We focus on noisy gates and modest initial fidelities — conditions that more accurately reflect current experimental platforms.  Our protocols can be optimized to tolerate arbitrary Pauli errors and achieve significant fidelity improvement without assuming idealized components. This shift in assumptions is crucial for enabling scalable GHZ-based architectures in large quantum networks.

\begin{figure}[!htbp]
    \centering
    \includegraphics[width=0.45\textwidth]{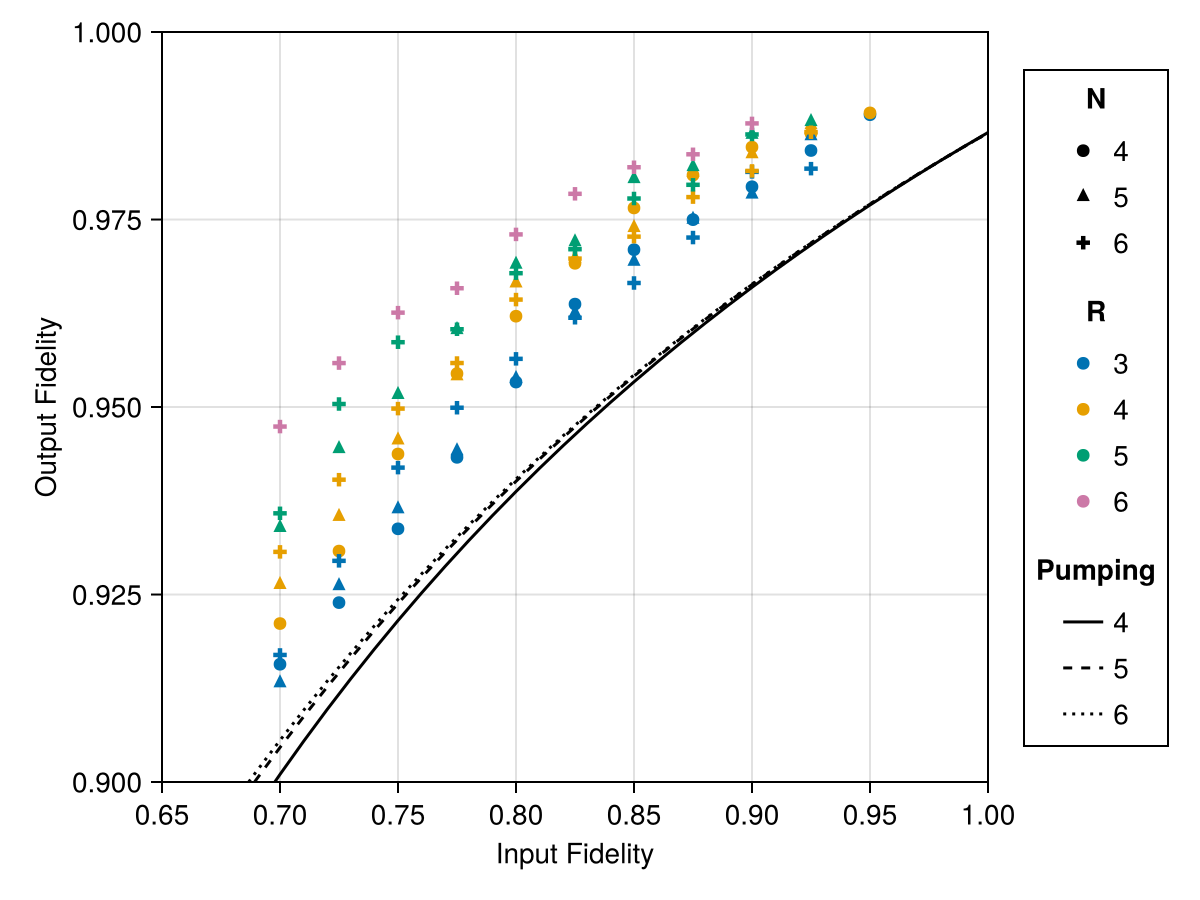}
    \includegraphics[width=0.45\textwidth]{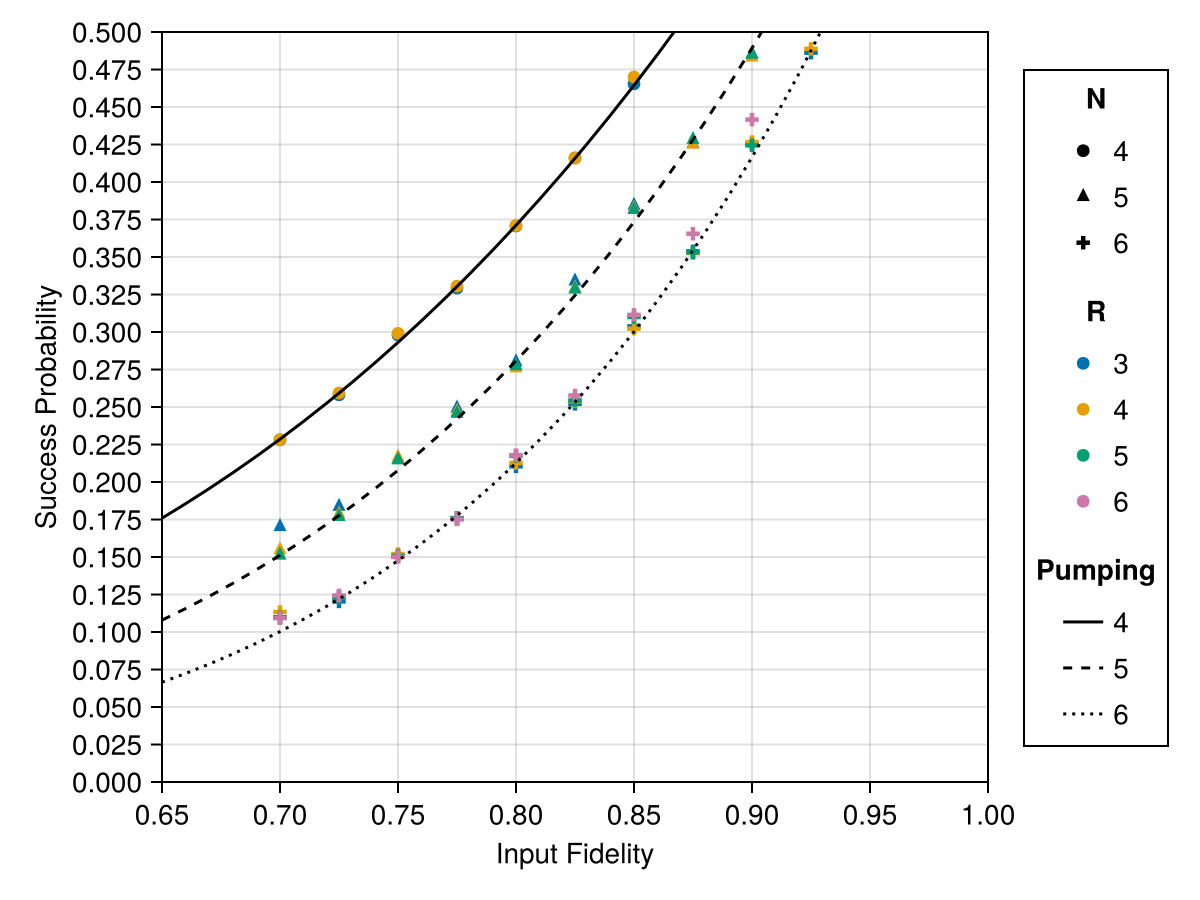}
    \caption{\textbf{Comparison between our optimized circuits and the recurrent pumping method.}
    Input fidelity vs output fidelity (top) and input fidelity vs success probability (bottom). 
    \textbf{N} denotes the number of raw GHZ states used by our circuits (encoded by marker shape); 
    \textbf{R} denotes the register size (how many qubits a node can store at a time, encoded by color). 
    Each point corresponds to a circuit obtained by our optimization under the given parameters.
    Solid, dashed, and dotted lines are the standard pumping baselines~(see Appendix~\ref{entanglement pumping}); the numbers listed under ``Pumping'' in the legend indicate the number of raw GHZ states used by the pumping protocol (solid: $N{=}4$, dashed: $N{=}5$, dotted: $N{=}6$). 
    All points are optimized to maximize the output fidelity given the lower bound of the pumping method’s success probability (gate error rate $p{=}0.01$, measurement error rate $\eta{=}0.01$). 
    {Each pumping curve is therefore compared with optimized circuits using the same total raw-state count $N$. Pumping requires only $R{=}2$; larger $R$ values indicate the additional register space available to the optimized circuits.}
    Our circuits achieve higher output fidelity for any given success probability.}
    \label{fig:twocolumn}
\end{figure}
\begin{figure*}[!tp]
    \centering
    \begin{minipage}{0.45\textwidth}
        \includegraphics[width=\textwidth]{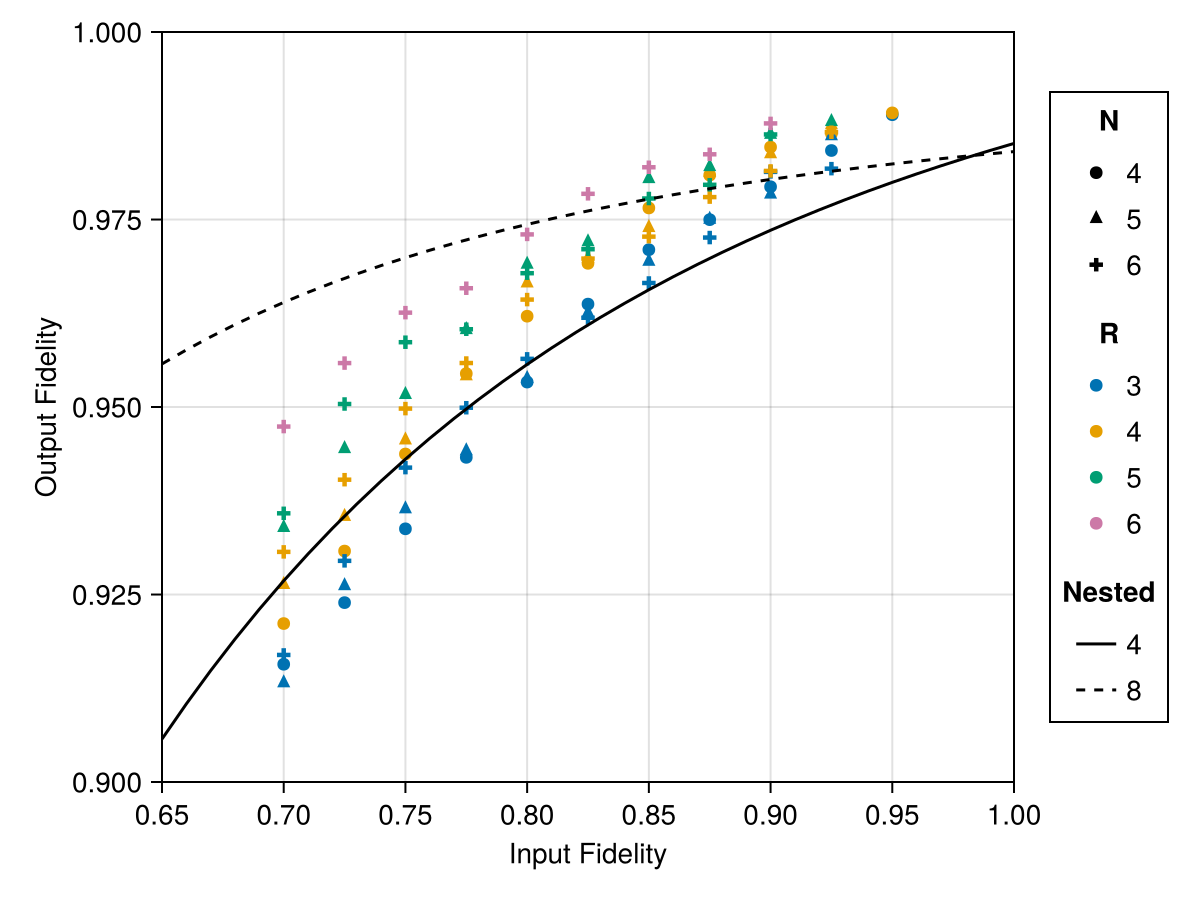}
        \vspace{1mm}
        \centering (a) Output fidelity
    \end{minipage}
    \hfill
    \begin{minipage}{0.45\textwidth}
        \includegraphics[width=\textwidth]{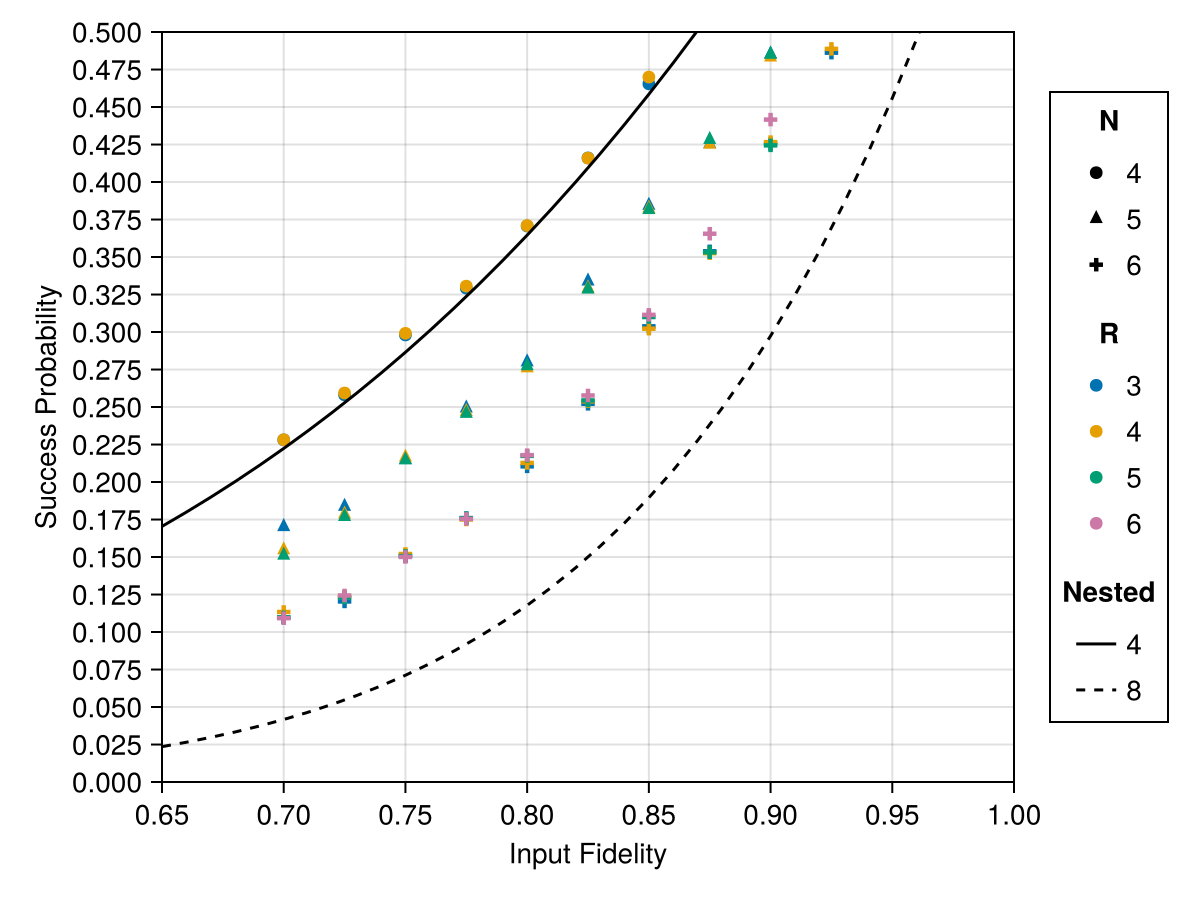}
        \vspace{1mm}
        \centering (b) Success probability
    \end{minipage}
    \caption{\textbf{Comparison between our optimized circuits and the nested distillation method.} 
    Input fidelity vs output fidelity (left) and input fidelity vs success probability (right), comparing against 2-round and 3-round nested protocols using 4 and 8 raw GHZ states, respectively. 
    Standard nested protocols typically assume twirling after each distillation round, where random Pauli gates are applied to symmetrize errors. This removes the strong $Z$-bias of the intermediate measurement outcomes and replaces it with an isotropic noise model (see Appendix~\ref{asymmetry}).
    In contrast, our circuits directly model and detect both $X$ and $Z$ errors without relying on twirling, and are optimized under realistic noise (gate error rate \(p=0.01\), measurement error rate \(\eta=0.01\)). As a result, we achieve higher output fidelities even with fewer raw GHZ states at moderate input fidelities. 
    {Under sequential execution, the $N{=}4$ and $N{=}8$ nested baselines require at least $R{=}3$ and $R{=}4$, respectively. The optimized results include the corresponding $R{=}3$ and $R{=}4$ cases; at moderate input fidelities, the $N{=}6$ circuits already outperform the $N{=}8$ nested baseline in both output fidelity and success probability, while $R{=}5$ and $R{=}6$ illustrate the smaller additional gain available from more register space.}
    Moreover, unlike nested protocols, which require exponentially growing and fixed numbers of raw states, our method supports arbitrary input counts, offering greater flexibility for circuit-level optimization. Both plots use the same data set as in Fig.~\ref{fig:twocolumn}, but are shown separately for clarity and to highlight different comparison baselines.}
    \label{fig:nest}
\end{figure*}

\subsection{Noise Modeling in GHZ States}
The raw input states are described as $n$-qubit GHZ states subject to isotropic noise
\begin{equation}
\rho_\text{in} \;=\; f_\text{in}\,\ket{\mathrm{GHZ}_n}\!\bra{\mathrm{GHZ}_n}
\;+\; (1-f_\text{in})\,\frac{I_n}{2^n}\,,
\end{equation}
where $f_\text{in}$ denotes the fidelity with respect to the ideal GHZ state and $I_n$ is the $n$-qubit identity operator.

In most examples below we use depolarizing noise, a standard and widely used error model equivalent to $X$, $Y$, and $Z$ Pauli errors with equal probability. Under depolarizing noise, the quantum state of a qubit undergoes the following transformation:
\begin{equation}
\rho \rightarrow (1 - p) \rho + \frac{p}{4} (\rho + X \rho X + Y \rho Y + Z \rho Z),   
\end{equation}
where \( \rho \) is the density matrix of the qubit, and \( X \), \( Y \), and \( Z \) represent the Pauli operators. This model captures the random application of bit-flip (\(X\)), phase-flip (\(Z\)), and combined (\(Y\)) errors, all of which can affect the qubit with equal probability.

We assume that every quantum gate in the circuit independently incurs depolarizing noise with probability \( p \), and each measurement will have probability $\eta$ to return an incorrect result. The goal of the circuit optimization described in the next subsection is to construct GHZ-preserving circuits that are robust to such noise and maximize the output fidelity of the final GHZ state.

Finally, we note that our framework naturally extends to more general biased Pauli noise models, in which $X$, $Y$, and $Z$ errors occur with unequal probabilities. These biased models include important physically motivated scenarios, such as amplitude damping (associated with $T_1$), which can be worst-case bounded by Pauli channels through Pauli twirling approximation~\cite{zhou2013}, while dephasing (associated with $T_2$) is already equivalent to a biased Pauli noise. All of these noise models are supported by our software implementation, allowing our approach to remain applicable in a broad range of experimental settings.

\subsection{Circuit Optimization for GHZ distillation}
Building on the group-decomposition framework introduced earlier, we leverage its \(O(1)\) complexity and reduced search space to efficiently search for and simulate a wide range of potential distillation circuits. In particular, the group decomposition and underlying group structure allow us to systematically construct GHZ-preserving circuits, greatly simplifying circuit generation, and enabling scalable optimization under realistic noise models. For the optimization we use a genetic algorithm (the sequence of gates being the genome of an individual circuit in a population of circuits), but that choice is not of significance.

In the various evaluations below, we denote the circuit configuration by four parameters:
\begin{itemize}
    \item \( N \): The number of raw GHZ states used in the distillation process.
    \item \( n \): The number of qubits in each GHZ state.
    \item \( K \): The desired number of output GHZ states after a successful distillation.
    \item \( R \): The size of the register used in the distillation circuit. We specifically permit nodes of small size that can not contain all $N$ qubits at the same time, requiring regeneration of raw GHZ states during the execution of the purification circuit. 
\end{itemize}

These parameters, along with the gate error rates \( p \), measurement error rates \( \eta \), and raw GHZ fidelity $f_\textrm{in}$ define the constraints under which the optimization is performed. 
{The numerical results assume independent GHZ-diagonal inputs, independent Pauli gate errors, independent measurement-bit errors, and fixed nonadaptive measurements and postselection. Resource counts include raw states and register capacity but exclude memory noise and storage time, communication cost and latency, and raw-state generation delay. Biased or correlated Pauli noise is supported; non-Pauli noise and the omitted dynamics are not evaluated here.}
{The input and hardware parameters} can be adjusted to reflect different hardware architectures, allowing our approach to adapt to a variety of experimental conditions. {For each configuration, we simulate circuit performance and use a genetic algorithm to explore the circuit space. The circuit representation, generation-dependent fitness evaluation, and settings used for the reported searches are specified in Appendix~\ref{app:cost_function}.}

Our optimized circuits achieve the following improvements:  
(1) \textbf{Higher output fidelity at the same success probability.}  
Against the recurrent-pumping baseline~(see Appendix~\ref{entanglement pumping}), this advantage holds uniformly across all tested input fidelities (see Fig.~\ref{fig:twocolumn}).  
Against the nested baseline, our circuits outperform at moderate-to-high input fidelities (see Fig.~\ref{fig:nest}); at low input fidelities, the nested protocol appears to perform better. This is because it includes twirling, a probabilistic operation that is not included in our optimizations (see also Appendix.~\ref{asymmetry}).  
(2) \textbf{Fewer raw GHZ states required.}  
In both comparisons, our circuits can reach a target output fidelity with fewer raw GHZ states; for nested protocols this reduction is most pronounced in the moderate-to-high input-fidelity regime. To give a concrete sense of what such optimized circuits look like, we show in Fig.~\ref{fig:generated circuit} an example generated by our algorithm.

This optimization process offers a scalable and efficient approach to GHZ-state distillation, capable of adapting to various hardware setups and error rates. By efficiently exploring the space of possible circuit designs, we can significantly reduce the computational complexity and achieve better performance in terms of both resource efficiency and output fidelity. 

\begin{figure}[!htbp]
    \centering
    \includegraphics[width=0.45\textwidth]{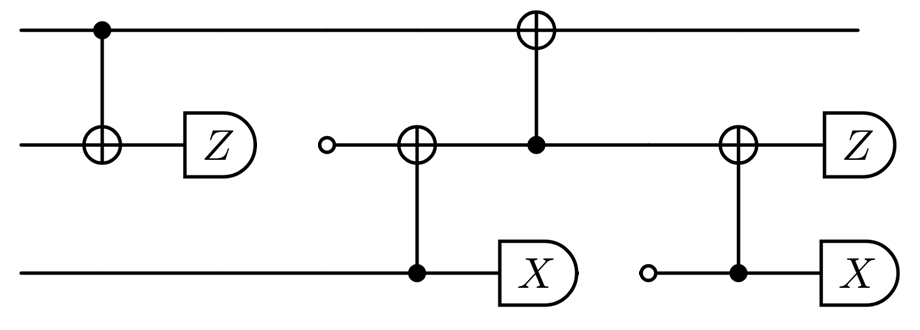}
    \caption{\textbf{Example of generated circuit.} This circuit is generated based on 5-to-1 3-qubit GHZ states, with constraint of register number $R=3$, gate error $p=0.01$, measurement error $\eta=0.01$, and input fidelity $f_{in}=0.9$. The small dot in the circuit means adding a new raw state after one is measured (register reuse). Notably, the figure only shows Alice's circuit, while Bob and Charlie apply the same operations on their respective registers. We deliberately selected this circuit for illustration because it only contains H-group gates, which ensures that all operations are applied homogeneously across all nodes. This homogeneity makes the circuit easier to visualize, as opposed to other generated circuits that may contain B-group gates, introducing bilocal operations between two nodes and making the diagram more complex to represent.}
    \label{fig:generated circuit}
\end{figure}

\section{Extension to Graph States}

While our decomposition framework is developed with GHZ states in mind, it naturally extends to a broader class of multipartite entangled states — namely, graph states. In future work we plan to explore the general case. In this section, we outline how our method can be applied directly to certain families of graph states without requiring substantial modification.

In particular, our approach is immediately applicable to star-shaped graph states, where a central qubit is entangled with all others via controlled-Z (CZ) gates. These states retain a symmetry structure that closely resembles GHZ states, and applying a single Hadamard gate to the central qubit of a star graph transforms it into a GHZ state, making them compatible with the same group decomposition framework. Moreover, the complete graph state can be transformed into a star-shaped graph through local complementation (LC) operations~\cite{hein2006}, allowing our methods to apply in this case as well. Examples of the necessary local adjustments are given in Table~\ref{tab:tran}.

A comprehensive treatment of arbitrary graph topologies and their associated preserving gate structures is beyond the scope of this paper and will be addressed in future work.

\begin{table}[htbp] 
\centering

A graph state:

\vspace{0.2cm}
\begin{tabular}{c
                c
                c}
\hline
X & Z & Z \\
Z & X & Z \\
Z & Z & X \\
\hline
\end{tabular}

\vspace{0.5cm}
Same graph state, different generating set of stabilizers:

\vspace{0.2cm}
\begin{tabular}{c
                c
                c}
\hline
X & X & X \\
Y & Y & I \\
I & Y & Y \\
\hline
\end{tabular}

\vspace{0.5cm}

The tableau after a change of basis gives us a GHZ state:

\vspace{0.2cm}
\begin{tabular}{c
                c
                c}
\hline
X & X & X \\
Z & Z & I \\
I & Z & Z \\
\hline
\end{tabular}

\vspace{0.5em}
\caption{\textbf{Local change of basis from a graph state to a GHZ state.} Standard stabilizer generators of a 3-qubit triangle graph state. \textbf{(b)} An alternative generator set for the same graph state. \textbf{(c)} Stabilizer generators for the 3-qubit GHZ state. The comparison between (a) and (b) demonstrates that different generator choices can represent the same graph state, while comparison between (b) and (c) showing a structural similarity between the graph and GHZ states. In particular, the generator \textbf{YYI} in (b) plays a similar role to \textbf{ZZI} in (c), suggesting that graph state operations can be viewed as GHZ transformations under a suitable basis adjustment.}
\label{tab:tran}

\end{table}

\section{Conclusion and Discussion}

In this work, we have introduced a complete framework for efficiently enumerating and simulating (in $O(1)$ time) GHZ-distillation operations, by decomposing local operations on pairs of GHZ states as products of operations from two small groups: the H group and the B group, defined in the main text. We have also extended this approach to certain graph states. This decomposition not only simplifies the design and simulation of GHZ-distillation circuits but also makes possible the creation of efficient tools for optimizing distillation circuits to bespoke hardware and noise models. We showcase how a simple optimizer using this drastically faster simulation technique in its cost function can easily produce distillation circuits much better than the state of the art. This is especially relevant for future quantum networks and distributed quantum computing, where high-fidelity entanglement and resource efficiency are paramount. Our approach has shown significant improvements in output fidelity, resource consumption, and tolerance to noise, making it a strong candidate for practical applications in real-world quantum networks and hardware.

Looking ahead, there are several promising avenues for future research. Much of this technique should be extendable to more general graph states, similarly to how early GHZ-purification techniques were extended to two-colorable graph states and then to arbitrary $k$-colorable graph states~\cite{aschauer2005,Kruszynska2006}, and we have early results that support this belief. More generally, this work can lead to improvements in measurement-based quantum computation (MBQC). Given the role of graph states in MBQC, applying our group decomposition framework to optimize graph state preparation and measurement protocols could lead to more efficient implementations of MBQC schemes. Another particularly interesting direction would be to incorporate these techniques into a dynamical model of entanglement generation, mimicking the behavior of a real quantum network rather than assuming on-demand, instantaneous access to raw states. Such an integration could reveal new trade-offs between fidelity, resource cost, and network timing constraints, and help bridge the gap between idealized protocols and practical distributed quantum architectures.

In conclusion, the group decomposition method we have introduced offers a scalable and efficient approach to simulating and optimizing quantum circuits for GHZ states and related graph states. Its flexibility and applicability to various quantum states make it a valuable tool for advancing quantum computation and communication. As quantum technologies continue to develop, we believe this method will play an increasingly important role in the design of robust, high-performance quantum networks and systems.

An open-source implementation of the $O(1)$ simulation technique described in this paper is available~\cite{githubghz}.

\section*{Author contributions}
Mingyuan Wang carried out the theoretical derivations, software and modeling work, and drafted the manuscript. Guus Avis guided the theoretical derivations. Kenneth Goodenough's critical feedback on the manuscript helped shape its revision. Stefan Krastanov supervised the project.

When revising the manuscript in response to peer review, we used OpenAI Codex (GPT-6 Astra) to polish the English in parts of the appendices. We also used it to help typeset the final manuscript and check the references.

\section*{Acknowledgments}
We acknowledge support from NSF grants 1941583, 2346089, 2402861, 2522101.

\bibliographystyle{quantum-doi}
\bibliography{bibliography}

\appendix

\section{Pauli and Clifford Groups}
The Pauli group $\mathcal{P}_n$ on $n$ qubits is generated by single-qubit Pauli operators $X_j=\begin{pmatrix}
  0&1 \\
  1&0
\end{pmatrix}$,
$Y_j=\begin{pmatrix}
  0&-i \\
  i&0
\end{pmatrix}$,
$Z_j=\begin{pmatrix}
  1&0 \\
  0&-1
\end{pmatrix}$ acting on the $j$th qubit, for $j = 1,\dots,n$.

Consider non-identity Pauli matrices $\mathcal{P}_n^*=\mathcal{P}_n \setminus {I^{\otimes n}}$
The Clifford Group $\mathcal{C}_n$ on $n$ qubits is 
\begin{equation}
\mathcal{C}_n=\{U \in U(2^n) | \sigma \in \pm \mathcal{P}_n^* \Rightarrow U \sigma U^\dagger \in \pm \mathcal{P}_n^*\}/U(1)
\end{equation}

The number of elements in $\mathcal{C}_n$ is
\begin{equation}
|\mathcal{C}_n|= \prod_{j=1}^{n}2(4^j-1)4^j=2^{n^2+2n}\prod_{j=1}^{n}(4^j-1)
\end{equation}

For any Clifford gate $U \in \mathcal{C}_n$, it is uniquely determined (up to a global phase) by how it maps each generator $X_i$ and $Z_i$ (where $i=1,2,\dots n$) through conjugation. This is because $X_i$ and $Z_i$ form a basis for the vector space of Pauli strings. When $U$ acts on the basis, it maps each $X$ and $Z$ such that the resulting images continue to satisfy the commutation and anti-commutation relations of the original operators~\cite{koenig2014, ozols2008, nebe2001}.

As an example, consider the Clifford group $\mathcal{C}_1$: $X$ can be mapped to any of $\pm \mathcal{P}_1^*$, but the image of $Z$ must anti-commute with the image of $X$ to preserve the structure of the Pauli algebra~\cite{ozols2008}.

We also introduce the "phaseless" Clifford group $\mathcal{C}^*_n=\mathcal{C}_n/\mathcal{P}_n$, the Clifford group on $n$ qubits modulo local Pauli operations $\mathcal{P}_n$, where $\mathcal{P}_n = \{\sigma_1 \otimes \cdots \otimes \sigma_n \mid \sigma_i \in \{I, X, Y, Z\}\}$. In other words $\mathcal{C}^*_n$ is a quotient group, i.e.\ each group element is itself a coset of multiple gates in $\mathcal{C}$, but all these gates have the same tableau entries and only differ by the phase column of the tableau (which is why we call it "phaseless"). The phaseless Clifford group $\mathcal{C}^*_n$ has further simple group form 
\begin{equation}
    \mathcal{C}^*_n \cong Sp(2n, \mathbb{F}_2) \equiv Sp(2n)
\end{equation}
~\cite{koenig2014}.

\section{Detailed Proofs}
\label{detailed_proofs}
Here we will introduce two sets of gates that act two GHZ states: $H$ (for homogeneous) and $B$ (for bilocal) -- they are both GHZ-preserving and phaseless (as defined below). We introduce them because they are conceptually convenient for classifying the group of all GHZ-preserving gates, and we will prove that they, together with the Pauli gates, generate the entire group of GHZ-preserving gates. It should be noted that we focus specifically on single-qubit and two-qubit gates, as our circuit construction is based on these fundamental operations and any larger gate can be decomposed into them.

\subsection{Stabilizers of the GHZ State}

In the stabilizer formalism, a quantum state is described as the unique joint $+1$ eigenstate of a set of commuting Pauli operators, known as the stabilizer generators. For the $n$-qubit GHZ state
\begin{equation}
| \mathrm{GHZ}_n \rangle = \frac{1}{\sqrt{2}} \left( |0\rangle^{\otimes n} + |1\rangle^{\otimes n} \right),   
\end{equation}
the stabilizer group is generated by:
\begin{itemize}
    \item The $X$-type stabilizer: $X^{\otimes n}$
    \item $n - 1$ $Z$-type stabilizers: $Z_i Z_{i+1}$ for $i = 1$ to $n-1$
\end{itemize}

For instance, the stabilizer generators of the 3-qubit GHZ state are:
\begin{equation}
\langle XXX,\ ZZI,\ IZZ \rangle.    
\end{equation}

The full stabilizer group consists of all $2^n$ Pauli strings generated by these operators (including their products).

This structure plays a central role throughout the following proofs, as GHZ-preserving gates are defined by their action on these stabilizers (up to phase).

\subsection{Qubit Indices}

The qubit at node $i$ in GHZ state $j$ will have an index $k = (i-1)\times n+j$, where indexing starts at 1, and $n$ is the number of qubits in a single GHZ state. In other words, we first enumerate the first qubit of each GHZ state (by listing each node) and then go to the second qubit of each, and so on.

Thus, for the Hilbert space of two GHZ states, if we write $g \otimes g \otimes I^{(2n-4)}$ for a two-qubit gate $g$, this corresponds to node 1 acting on its two qubits (each being a part of a different GHZ state) with the gate $g$, and node 2 acting on its two qubits with the same gate, while everyone else does not act on their qubits.

\subsection{Definitions}
\label{sec:definitions}
We begin by clarifying several important terms used throughout this work:

\paragraph*{Number of qubits ($n$).}
We consider GHZ states of $n$ qubits distributed among $n$ nodes (e.g., Alice, Bob, Charlie, etc.), with each node holding exactly one qubit from each GHZ state.
{The classification results below assume $n\geq 3$. The case $n=2$ is the bipartite Bell-state setting and has a different local symmetry group.}

\paragraph*{GHZ basis.}
Any state that is represented by the same stabilizer tableau (up to differences in signs, i.e.\ up to being different by local Pauli operators) is within what we call the \emph{GHZ basis}. Equivalently, these basis states can be generated from $\lvert \mathrm{GHZ} \rangle$ by applying local Pauli operations. There are $2^n$ such states in the GHZ basis for $n$ qubits.

\paragraph*{GHZ-preserving.}
An operation (or quantum gate) is called \emph{GHZ-preserving} if, for any input state from the GHZ basis, the output remains within the GHZ basis. Concretely, if the input is an $n$-qubit GHZ-basis state, the output must still be one of the $2^n$ GHZ-basis states.

\paragraph*{Bilocal.}
In a network of $n$ qubits (one per node), an operation is said to be \emph{bilocal} if it acts on exactly two qubits (for instance, one qubit held by Alice and one qubit held by Bob), while acting as the identity on all remaining qubits. In the case of a two-qubit gate (acting on qubits from two separate GHZ states), a bilocal gate is applied to a pair of qubits within each of two different nodes. Note that in this definition, we do not require Alice and Bob to apply the same gate on their respective qubits. That restriction will arise naturally in the discussion of the more specific B group of gates below.

\paragraph*{Phaseless.}
For much of our analysis, we work within the \emph{phaseless Clifford group}, denoted $C_2^* = C_2 / P_2$, the quotient group of the two-qubit Clifford group $C_2$ by the local Pauli group $P_2$. That is, two Clifford gates are considered equivalent in $C_2^*$ if they differ only by a local Pauli operation.
Each element of $C_2^*$ can be represented either as a coset of gates (i.e., $C_2$ gates that are equivalent up to a Pauli), or equivalently, as a single "phaseless" $C_2$ gate. In a given coset, all gates have the same stabilizer tableau representation, differing only in the phase column. The "phaseless" $C_2$ gate is the representative in the coset whose tableau has only positive phases. In the "phaseless" representation, the group operation is still standard gate multiplication, but phases of the stabilizer operators (the rows of the tableau) are not tracked and are assumed to be positive. This representation is useful for classifying all operations.


\medskip
With these definitions in place, we now proceed to examine the structure of the two "building block" groups that let us enumerate all GHZ-preserving two qubit gates (multi-qubit gates can be decomposed into two-qubit gates).

\subsection{The B Group} 
\label{Appendix:B}
In this subsection, we focus on bilocal gates, i.e.\ which act at exactly two nodes (here taken to be nodes 1 and~2 without loss of generality). Among all GHZ-preserving gates, we define the $B$ group as the subgroup consisting of those that are bilocal and phaseless. We will enumerate all such gates in two steps: first considering tensor products of single-qubit gates, and then examining two-qubit gates. We also show that the resulting group of bilocal, phaseless, GHZ-preserving operations is isomorphic to $\mathbb{Z}_2 \otimes \mathbb{Z}_2 \otimes \mathbb{Z}_2$.

Recall the definition of the $B$ group:
\newtheorem*{definition*}{Definition}
\begin{definition*}
\leavevmode
\begin{enumerate}
    \item \textbf{GHZ-preserving:} Each gate maps GHZ-basis state to GHZ-basis state.
    \item \textbf{Phaseless:} A gate applied at a given node belongs
 to the quotient group $\mathcal{C}^*_2$
    \item \textbf{Bilocal:} Only two nodes apply gates (not necessarily the same gate).
\end{enumerate}
\end{definition*}

\begin{lemma}[Single-qubit gates in the GHZ-preserving group]
    Only the Phase gate ($S$) can appear as a single-qubit operation in the GHZ-preserving group.
    \label{lemma:singlequbitB}
\end{lemma}

\begin{table*}[!tp]
\centering
\begin{tabular}{|c|c|c|c|}
\hline
XI $\to$ + XI & XI $\to$ + XI & XI $\to$ + YI & XI $\to$ + YI \\
IX $\to$ + IX & IX $\to$ + IY & IX $\to$ + IX & IX $\to$ + IY \\
ZI $\to$ + ZI & ZI $\to$ + ZI & ZI $\to$ + ZI & ZI $\to$ + ZI\\
IZ $\to$ + IZ & IZ $\to$ + IZ & IZ $\to$ + IZ & IZ $\to$ + IZ\\
\hline
\makebox[3cm][c]{$Identity$} & \makebox[3cm][c]{$I \otimes S$} & \makebox[3cm][c]{$S \otimes I$} & \makebox[3cm][c]{$S \otimes S$}

\end{tabular}
\begin{tabular}{|c|c|c|c|}
\hline
XI $\to$ + XZ & XI $\to$ + XZ & XI $\to$ + YZ & XI $\to$ + YZ \\
IX $\to$ + ZX & IX $\to$ + ZY & IX $\to$ + ZX & IX $\to$ + ZY \\
ZI $\to$ + ZI & ZI $\to$ + ZI & ZI $\to$ + ZI & ZI $\to$ + ZI \\
IZ $\to$ + IZ & IZ $\to$ + IZ & IZ $\to$ + IZ & IZ $\to$ + IZ \\
\hline
\makebox[3cm][c]{$CZ$} & \makebox[3cm][c]{$CZ \cdot I \otimes S$} & \makebox[3cm][c]{$CZ \cdot S \otimes I$} & \makebox[3cm][c]{$CZ \cdot S \otimes S$} \\
\hline
\end{tabular}
\caption{Copy of Table~\ref{tab:Fgates}. B group gates and their corresponding mappings}
\label{tab:Fgates_sup}
\end{table*}

\begin{proof}
    Consider a three-qubit GHZ state with stabilizers \( XXX \), \( ZZI \), \( ZIZ \), \( IZZ \), \( YYX \), \( YXY \), \( XYY \), and \( III \). Since we work within the Clifford group (more precisely, the phaseless Clifford group), single-qubit gate can only map one of Pauli operators (\( X, Y, Z \)) to another Pauli operators. Importantly, any stabilizer that includes \( I \) (e.g., \( ZZI, ZIZ, IZZ \)) must be mapped to another stabilizer that also includes \( I \) in the same position. This restriction arises because a single-qubit gate cannot introduce or remove \( I \). 
    As a result, \( ZZI \) must be mapped to \( ZZI \), \( ZIZ \) to \( ZIZ \), and \( IZZ \) to \( IZZ \). This implies that \( Z \) must be mapped to \( Z \) on every qubit, and no other Pauli (such as \( X \) or \( Y \)) can appear in its place. 
    With $Z$ required to map to $Z$, we are left with the transformations of $X$ and $Y$. The potential mappings are either $X \rightarrow X$, $Y \rightarrow Y$ (Identity) or $X \rightarrow Y$, $Y \rightarrow X$ (Phase gate, $S$). Thus, when considering single qubit gates, Alice and Bob are not permitted to act with anything besides $S$ gates. The same argument holds for any $n$-qubit GHZ state, as its $Z$-type stabilizers are always of the form $Z_i Z_{i+1}$.
\end{proof}

\begin{lemma}[Bilocal requirement]
    The gates $S\otimes I \otimes S\otimes I \otimes I^{2n-4}$ and $I\otimes S \otimes I\otimes S \otimes I^{2n-4}$, where $S$ is the single-qubit phase gate, generate all phaseless GHZ-preserving gates that act bilocally on the first two nodes that are factorizable in single-qubit gates.
    \label{lemma:bilocal}
\end{lemma}

Note: above we specifically set the statement to apply to the nodes sharing two GHZ states, because in future lemmas we want to start considering two-qubit gates.

\begin{proof}
    Now we will further constrain how the $S$ gate is applied, showing that if Alice acts on qubit $i$ then Bob has to act on qubit $i+1$, i.e.\ if Alice acts on her qubit of a given GHZ state with the $S$ gate, then Bob has to act on his qubit of the same GHZ state.

    The GHZ state for three qubits includes the stabilizers ($XXX$, $YYX$, $YXY$, $XYY$). The Phase gate S affects these stabilizers by transforming $X$ to $Y$ and $Y$ to $X$ while keeping $Z$ unchanged. This transformation permutes some stabilizers into each other.
    
    For an \( n \)-qubit GHZ state, there are \( 2^{n-1}-1 \)  Z-containing stabilizers (e.g., \( ZZI, ZIZ, IZZ \) for \( n=3 \)), each containing an even number of \( Z \) operators. Correspondingly, there are \( 2^{n-1}-1 \) $Y$-containing stabilizers, each containing even number of $Y$. If Alice acts on a qubit of a GHZ state with the $S$ gate, then we will obtain a state whose set of stabilizers includes a stabilizer containing exactly one $Y$, which is not permitted. Thus Bob also has to act on his qubit of the corresponding GHZ state to remain in the GHZ basis (i.e.\ $XXX$ $\to$ $YYX$). Conversely, as long as the $S$ gates appear in pairs, since every $Z$-containing stabilizer of the GHZ state contains two $Z$ operators, one can always find a corresponding stabilizer generator with an even number of $Y$ operators. Hence paired $S$ gates are sufficient to remain within the GHZ basis. See Fig.~\ref{circuit phase gate}.
\end{proof}

\begin{figure}[!htbp]
    \centering
        \begin{quantikz}
        \lstick[2]{Alice} & \gate[1]{S} & \qw &  \\
        \lstick{} &&&  \\
        \lstick[2]{Bob} &  \gate[1]{S} &  \qw & \\
        \lstick{} &&&  \\
        \lstick[2]{Charlie} &  \qw &  \qw & \\
        \lstick{} &&& \\
        \end{quantikz}
    \caption{Bilocal Phase Gates}
    \label{circuit phase gate}
\end{figure}
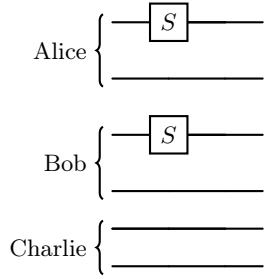

\begin{lemma}[Constraints on two-qubit gates in the $B$ group]
    Let $G = g_1 \otimes g_2 \otimes I^{\otimes 2n-4} \in B$, where $g_1$ and $g_2$ are two-qubit gates. Then both $g_1$ and $g_2$ must map a $Z$ entry of a GHZ stabilizer to a $Z$, and the only non-trivial such gate is the Controlled-$Z$ (CZ) gate, possibly composed with single-qubit phase ($S$) gates as described in Lemma~\ref{lemma:singlequbitB}.
    \label{lemma:czB}
\end{lemma}

\begin{proof}
    (i) Consider one of the two GHZ states and one of its stabilizer operators, namely $ZZI\dots I$. The gate $G$ does not act on the qubits past the second one, thus all of the $I$ entries have to remain $I$. The only stabilizer operators fulfilling this constraint is $ZZI\dots I$, as any stabilizer generators contains $X$ or $Y$ have full weights, thus the gate $g$ has to map a $Z$ entry to a $Z$ entry (as a reminder, we are still consider only phaseless gates). Permitted gates of such type are exhaustively enumerated in Table~\ref{tab:allowZsup}. Of note, further down we will see that only a subset of these gates is actually in $B$, due to other constraints.
    
    (ii) Exhaustively checking the 36 pairings of the 6 entries in Table~\ref{tab:allowZsup}, shows us that only the first entry of the table preserves GHZ states when applied bilocally (i.e.\ by Alice and Bob at the same time). That entry corresponds to $Z$-mapping of the CZ gate (the $X$ mapping is not constrained yet). The freedom on the $X$ part of the tableau corresponds to the application of the single-qubit gates permitted by Lemma~\ref{lemma:singlequbitB}.
\end{proof}

\begin{figure}[!htbp]
    \centering
        \begin{quantikz}
        \lstick[2]{Alice} & \gate[2]{CZ} & \qw &  \\
        \lstick{} &&&  \\
        \lstick[2]{Bob} &  \gate[2]{CZ} &  \qw & \\
        \lstick{} &&\gate[1]{S}&  \\
        \lstick[2]{Charlie} &  \qw &  \qw & \\
        \lstick{} && \gate[1]{S} & \\
        \end{quantikz}
    \caption{Example of two B group application, first is CZ gate at Alice and Bob, second is S gate at Bob and Charlie}
    \label{circuit CZ gate}
\end{figure}
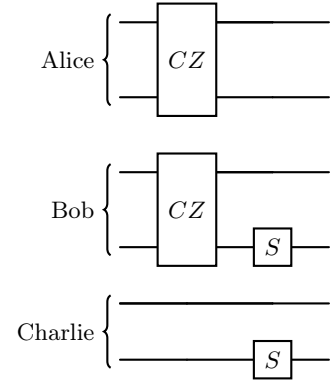

\begin{table*}[!tp]
\centering
\begin{tabular}{|c|c|c|c|c|c|}
\hline
ZI $\to$ + ZI & ZI $\to$ + IZ & ZI $\to$ + ZI & ZI $\to$ + ZZ & ZI $\to$ + IZ & ZI $\to$ + ZZ \\
IZ $\to$ + IZ & IZ $\to$ + ZI & IZ $\to$ + ZZ & IZ $\to$ + ZI & IZ $\to$ + ZZ & IZ $\to$ + IZ \\
\hline
\end{tabular}
\caption{Allowed Z mapping for GHZ-preserving. An underscore corresponds to the identity operator (typographic choice of ours to help with legibility).}
\label{tab:allowZsup}
\end{table*}

\begin{table*}[!tp]
\centering
\begin{tabular}{|c|c|c|c|c|c|}
\hline
XI $\to$ + XI & XI $\to$ + IX & XI $\to$ + XX & XI $\to$ + IX & XI $\to$ + XX & XI $\to$ + XI \\
IX $\to$ + IX & IX $\to$ + XI & IX $\to$ + IX & IX $\to$ + XX & IX $\to$ + XI & IX $\to$ + XX \\
ZI $\to$ + ZI & ZI $\to$ + IZ & ZI $\to$ + ZI & ZI $\to$ + ZZ & ZI $\to$ + IZ & ZI $\to$ + ZZ \\
IZ $\to$ + IZ & IZ $\to$ + ZI & IZ $\to$ + ZZ & IZ $\to$ + ZI & IZ $\to$ + ZZ & IZ $\to$ + IZ \\
\hline
\makebox[2cm][c]{$Identity$} & \makebox[2cm][c]{$SWAP$} & \makebox[2cm][c]{$CNOT_{12}$} & \makebox[2cm][c]{$DCX_{21}$} & \makebox[2cm][c]{$DCX_{12}$} & \makebox[2cm][c]{$CNOT_{21}$} \\
\hline
\end{tabular}
\caption{Copy of Table~\ref{tab:Hgates}. H group gates and their corresponding mappings.}
\label{tab:Hgates_sup}
\end{table*}

\begin{theorem}
    The $B$ group is generated by the phase gate (\(S\)) and the Control-Z (\(CZ\)) gate applied identically bilocally, e.g., the $B$ group on nodes 1 and 2 is
    \begin{align}
    B_{12} = \{\, & (g\cdot(f_1 \otimes f_2)) \otimes (g \cdot(f_1 \otimes f_2)) \otimes I^{2n-4} \mid \notag \\
            & g \in \{I, CZ\},\ (f_1, f_2) \in \{I, \mathit{S}\}^2 \,\}
    \end{align}

    \label{theorem:Bgroup}
\end{theorem}
\begin{proof}
    The theorem follows directly from Lemma~\ref{lemma:bilocal} and Lemma~\ref{lemma:czB}, and from the fact that the S gate and CZ gate commute.
\end{proof}

\begin{corollary}
All gates in the $B$ group need to be applied "identically" bilocally, i.e.\ the same gate is applied at both locations.
If a gate $G=g_1\otimes g_2 \otimes I^{\otimes 2n-4} \in B$, then $g_1=g_2$.
\end{corollary}
\begin{proof}
    This directly follows from Theorem~\ref{theorem:Bgroup}.
\end{proof}

In other words, if a phaseless operation on two GHZ states is bilocal, it is guaranteed to be made of the repetition of the same two phaseless two-qubit gates.

{
For use below, let
\begin{equation}
\mathcal{B}_0
=
\left\{
CZ^{\gamma}(S\otimes I)^{\alpha}(I\otimes S)^{\beta}
\;\middle|\;
\alpha,\beta,\gamma\in\{0,1\}
\right\}
\end{equation}
denote the eight local two-qubit factors listed in Table~\ref{tab:Fgates_sup}. In the phaseless quotient, these three generators commute and square to the identity, so $\mathcal{B}_0\cong\mathbb{Z}_2^3$. An element $B_{ij}(r)$ applies the same factor $r\in\mathcal{B}_0$ at nodes $i$ and $j$.

\begin{theorem}
For an $n$-qubit GHZ state with $n\geq3$, every product of bilocal $B$-group operations has a unique representation using the path edges $(1,2),(2,3),\ldots,(n-1,n)$, with one element of $\mathcal{B}_0$ assigned to each edge. Identity edge factors are allowed.
\label{theorem:n-1}
\end{theorem}
\begin{proof}
Consider an arbitrary product $\prod_{\{i,j\}\in P}B_{ij}(q_{ij})$, where $P$ is the set of node pairs carrying a nonidentity edge factor. Because $\mathcal{B}_0$ is abelian, its total action at node $i$ is the local factor
\begin{equation}
r_i=\prod_{j:\,\{i,j\}\in P}q_{ij}\in\mathcal{B}_0.
\end{equation}
Every edge factor occurs at both endpoints. Hence
\begin{equation}
\prod_{i=1}^{n}r_i
=
\prod_{\{i,j\}\in P}q_{ij}^{\,2}
=I.
\label{eq:B_node_constraint}
\end{equation}

Conversely, suppose $r_1,\ldots,r_n\in\mathcal{B}_0$ satisfy Eq.~\eqref{eq:B_node_constraint}. Define the path-edge factors
\begin{equation}
q_i=\prod_{j=1}^{i}r_j,
\qquad i=1,\ldots,n-1,
\label{eq:B_path_inverse}
\end{equation}
and apply $B_{i,i+1}(q_i)$ on edge $(i,i+1)$. Node $1$ receives $q_1=r_1$. Each intermediate node receives
\begin{equation}
q_{i-1}q_i=r_i,
\end{equation}
because every element of $\mathcal{B}_0$ is self-inverse. The last node receives
\begin{equation}
q_{n-1}=\prod_{j=1}^{n-1}r_j=r_n,
\end{equation}
where the last equality follows from Eq.~\eqref{eq:B_node_constraint}. Equation~\eqref{eq:B_path_inverse} also recovers every edge factor uniquely from the node factors. Thus the path representation is both sufficient and non-redundant.
\end{proof}
}

\subsection{The H Group}
\label{Appendix:H}
We now turn to the class of \emph{homogeneous} gates, where each of the $n$ nodes applies the \emph{same} two-qubit operation on its qubits (each of which belongs to a different GHZ state). We start with a group generated by the CNOT gate alone, and later on we will show that this is enough to generate any phaseless GHZ-preserving gate. We call this group the \emph{$H$ group}. We also show $H$ is isomorphic to the dihedral group $D_3$, the symmetry group of the equilateral triangle.

\newtheorem*{definition**}{Definition}
\begin{definition**}
\leavevmode
$H$ group is a group generated by \(\mathit{CNOT}_{12}^{\otimes n}\) and \(\mathit{CNOT}^{\otimes n}_{21}\)(up to local Pauli operations).
\end{definition**}

\begin{lemma}[Enumerating the six elements in $H$]
There are exactly six elements in $H$.
\end{lemma}

\begin{proof}
First, observe that each node applying $\mathit{CNOT}_{12}$ (or $\mathit{CNOT}_{21}$) indeed maps a GHZ-basis state to a GHZ-basis state. 
We enumerate all possible compositions and verify that they form a closed set under multiplication. The group consists of exactly six elements. See Table~\ref{tab:Hgates_sup}.
\end{proof} 

\begin{figure}[!htbp]
        \begin{quantikz}
        \lstick[2]{Alice} & \ctrl{1} &  &  &  & \\
        \lstick{} &\targ{}&& \qw & \qw & \\
        \lstick[2]{Bob} & \ctrl{1} &  &  &  & \\
        \lstick{} &\targ{}&& \qw & \qw & \\
        \lstick[2]{Charlie} & \ctrl{1} &  &  & & \\
        \lstick{} &\targ{}&& \qw & \qw & \\
        \end{quantikz}
        \caption{Example of H group application, tensor product of $CNOT_{12}$}
        \label{circuits_H}
\end{figure}
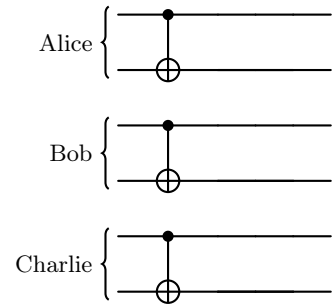

\subsection{Completeness of B and H}

Let $\mathcal{G}_n$ be the set of all phaseless local GHZ-preserving gates acting on two $n$-qubit GHZ states (distributed among $n$ nodes).  

{
\begin{lemma}[Eight residual mappings for a fixed $Z$-mapping]
\label{lem:eight_residual}
Let $r\in\mathcal{C}_2^*$ fix both $ZI$ and $IZ$. There are exactly eight possible images of $XI$ and $IX$ that preserve the canonical commutation and anticommutation relations, and these eight operations are precisely the elements of $\mathcal{B}_0$.
\end{lemma}
\begin{proof}
Because $XI$ anticommutes with $ZI$ and commutes with $IZ$, its image must have an $X$ or $Y$ on the first qubit and an $I$ or $Z$ on the second. Up to phase, it therefore has the form
\begin{equation}
r(XI)=(XZ^a)\otimes Z^b,
\qquad a,b\in\mathbb{F}_2.
\label{eq:residual_XI}
\end{equation}
Similarly, $IX$ commutes with $ZI$ and anticommutes with $IZ$, so
\begin{equation}
r(IX)=Z^c\otimes(XZ^d),
\qquad c,d\in\mathbb{F}_2.
\label{eq:residual_IX}
\end{equation}
The operators $XI$ and $IX$ commute, and their images must also commute. In Eqs.~\eqref{eq:residual_XI} and~\eqref{eq:residual_IX}, the first tensor factors anticommute exactly when $c=1$, while the second tensor factors anticommute exactly when $b=1$. The two Pauli strings commute if the total number of local anticommutations is even. Hence
\begin{equation}
b+c=0\pmod 2,
\qquad\text{or equivalently}\qquad c=b.
\label{eq:residual_commutation}
\end{equation}
The three bits $a,b,d$ remain independent, giving $2^3=8$ mappings. They are generated by $S\otimes I$, $CZ$, and $I\otimes S$, respectively. Conversely, every assignment of $a,b,d$ in Eqs.~\eqref{eq:residual_XI}--\eqref{eq:residual_commutation} preserves all canonical commutation relations and therefore determines a unique phaseless Clifford operation. These are exactly the eight elements of $\mathcal{B}_0$.
\end{proof}
}

\begin{theorem}[The $B$ and $H$ groups generate a subgroup of $\mathcal{G}_n$ of size $6\,\times\,8^{\,n-1}$]
\label{thm:BH-subgroup}

{
By Theorem~\ref{theorem:n-1}, every product of $B$ operations has a unique path normal form. Thus
\begin{equation}
    \langle B, H\rangle
\;=\;
\biggl\{\,
  h \,\prod_{i=1}^{n-1} B_{i,i+1}(q_i)
  \;\bigg|\;
  h \in H,~
  q_i\in\mathcal{B}_0
\biggr\}.
\end{equation}
}

Then:

\begin{enumerate}
\item All these gates are GHZ-preserving,~i.e.~$\langle B,H\rangle \subseteq \mathcal{G}_n$

\item The total number of distinct (phaseless) gates in $\langle B,H\rangle$ is
    \begin{equation}
    6 \;\times\; 8^{\,n-1}.    
    \end{equation}

\end{enumerate}
\end{theorem}

{
\begin{proof}
Every $h\in H$ and every $B_{ij}(r)$ is GHZ-preserving, and this property is closed under composition. Hence $\langle B,H\rangle\subseteq\mathcal{G}_n$.

It remains to count the normal forms without assuming that different choices are distinct. Each factor in $\mathcal{B}_0$ fixes $ZI$ and $IZ$, so the common $Z$-mapping of
\begin{equation}
h\prod_{i=1}^{n-1}B_{i,i+1}(q_i)
\end{equation}
is determined entirely by $h$. The six elements of $H$ have the six distinct $Z$-mappings listed in Table~\ref{tab:Hgates_sup}. Therefore, if two normal forms define the same phaseless operation, their $H$ factors must be equal.

After removing this common $H$ factor, the remaining product has well-defined local factors $r_1,\ldots,r_n\in\mathcal{B}_0$. By Theorem~\ref{theorem:n-1}, its path-edge factors are recovered uniquely as
\begin{equation}
q_i=\prod_{j=1}^{i}r_j.
\end{equation}
Thus equality of the phaseless operations forces equality of every path-edge factor; no two different $H/B$ assignments collide. There are six choices of $h$ and eight independent choices on each of the $n-1$ path edges, including the identity. Consequently,
\begin{equation}
|\langle B,H\rangle|=6\times8^{\,n-1}.
\end{equation}
\end{proof}
}

\begin{theorem}[Counting of GHZ-Preserving Gates]
\label{thm:ghz_preserving_enumeration}
Recall the $\mathcal{G}_n$ is the set of all phaseless Clifford unitaries (i.e., modulo local Pauli equivalences) acting on \emph{two} $n$-qubit GHZ-basis states, such that every GHZ-basis state is mapped to another GHZ-basis state.  Then
\begin{equation}
    \bigl|\mathcal{G}_n\bigr| \;\leq\; 6 \,\times\, 8^{\,n-1}.
\end{equation}

\end{theorem}

{
\begin{proof}
Let $G=g_1\otimes\cdots\otimes g_n\in\mathcal{G}_n$, where $g_i\in\mathcal{C}_2^*$ acts on the two qubits held at node $i$.

\noindent\textbf{1. The common $Z$-mapping.}
The $Z$-type stabilizers of two GHZ states are generated by adjacent products of $ZI$ and of $IZ$. Because $n\geq3$, a Pauli string with an $X$ or $Y$ entry at only two nodes cannot be a GHZ stabilizer: every GHZ stabilizer containing $X$ or $Y$ has support at all $n$ nodes. Thus each $g_i$ maps $ZI$ and $IZ$ into the nonzero $Z$-type operators
\begin{equation}
\{ZI,IZ,ZZ\}.
\end{equation}
The two images must be distinct because Clifford conjugation is invertible. There are therefore $3\times2=6$ possible ordered $Z$-mappings, as listed in Table~\ref{tab:allowZsup}.

Now consider the stabilizers obtained by placing the same operator, either $ZI$ or $IZ$, at adjacent nodes $i$ and $i+1$. Their images have support only at those two nodes. In each output GHZ copy, a purely $Z$-type stabilizer has even $Z$ parity. The image at node $i$ must therefore equal the image at node $i+1$. Applying this argument along the chain gives
\begin{equation}
g_1|_Z=g_2|_Z=\cdots=g_n|_Z.
\label{eq:common_Z_mapping}
\end{equation}
The six elements of $H$ realize these six common mappings uniquely. Let $h\in H$ be the corresponding homogeneous gate and define the residual local operations
\begin{equation}
r_i=h^{-1}g_i.
\end{equation}
Every $r_i$ fixes $ZI$ and $IZ$, so Lemma~\ref{lem:eight_residual} gives $r_i\in\mathcal{B}_0$.

\noindent\textbf{2. The global $X$-type constraints.}
Write the residual action at node $i$ as
\begin{align}
r_i(XI)&=(XZ^{a_i})\otimes Z^{b_i},\label{eq:global_residual_XI}\\
r_i(IX)&=Z^{b_i}\otimes(XZ^{d_i}),\label{eq:global_residual_IX}
\end{align}
where $a_i,b_i,d_i\in\mathbb{F}_2$. The equality of the two cross exponents in these expressions is exactly the commutation constraint derived in Eq.~\eqref{eq:residual_commutation}.

Apply Eq.~\eqref{eq:global_residual_XI} to the global $X$-type stabilizer of the first GHZ copy. Its image contains $X$ on every qubit of the first copy, together with $Z$-decorations of parity $\sum_i a_i$ on that copy and $\sum_i b_i$ on the second copy. Such an operator belongs to the two-copy GHZ stabilizer only if both parities are even:
\begin{equation}
\sum_{i=1}^{n}a_i=0,
\qquad
\sum_{i=1}^{n}b_i=0
\pmod 2.
\label{eq:first_global_X_constraint}
\end{equation}
Applying Eq.~\eqref{eq:global_residual_IX} to the second global $X$-type stabilizer similarly gives
\begin{equation}
\sum_{i=1}^{n}b_i=0,
\qquad
\sum_{i=1}^{n}d_i=0
\pmod 2.
\label{eq:second_global_X_constraint}
\end{equation}
The three independent conditions in Eqs.~\eqref{eq:first_global_X_constraint} and~\eqref{eq:second_global_X_constraint} are equivalent to
\begin{equation}
\prod_{i=1}^{n}r_i=I
\qquad\text{in }\mathcal{B}_0.
\label{eq:global_residual_product}
\end{equation}
Consequently, $r_1,\ldots,r_{n-1}$ may be chosen independently from the eight elements of $\mathcal{B}_0$, while the final residual mapping is fixed uniquely by
\begin{equation}
r_n=\left(\prod_{i=1}^{n-1}r_i\right)^{-1}.
\label{eq:last_residual_fixed}
\end{equation}
There are six choices for the common $Z$-mapping and at most $8^{n-1}$ compatible residual assignments. Hence
\begin{equation}
|\mathcal{G}_n|\leq6\times8^{\,n-1}.
\end{equation}
\end{proof}
}

\begin{theorem}[Completeness of B and H for GHZ-preserving gates]
\label{thm:completeness_BH}

Let $G$ be the group of all phaseless Clifford unitaries acting on two $n$-qubit GHZ-basis states such that any GHZ-basis state is mapped to another GHZ-basis state (i.e.\ the GHZ-preserving group). 
Then
\begin{equation}
  G \;=\;\langle B,H\rangle,  
\end{equation}

and in particular,
\begin{equation}
 \bigl|G\bigr|
  \;=\;
  6\,\times\,8^{\,n-1}.      
\end{equation}

\end{theorem}

\begin{proof}
In Theorem~\ref{thm:BH-subgroup}, we showed that 
\begin{equation}
\langle B,H\rangle 
\;\subseteq\; 
G
\quad\text{and}\quad
\bigl|\langle B,H\rangle\bigr|
\;=\;
6 \;\times\; 8^{n-1}.  
\end{equation}

In Theorem~\ref{thm:ghz_preserving_enumeration}, we also showed that 
\begin{equation}
  |G|
\;\le\;
6 \;\times\;8^{\,n-1}.  
\end{equation}

Since $\langle B,H\rangle$ is a subgroup of $G$ (both being finite groups) and they have the same finite cardinality, it follows that
\begin{equation}
\langle B,H\rangle 
\;=\;
G
\quad\text{and}\quad
|G|
\;=\;
6 \;\times\;8^{\,n-1}.    
\end{equation}

{
The preceding proofs give a unique phaseless normal form for every GHZ-preserving gate:
\begin{equation}
G_n
=
\left\{
h\prod_{i=1}^{n-1}B_{i,i+1}(q_i)
\;\middle|\;
h\in H,\ q_i\in\mathcal{B}_0
\right\}.
\end{equation}
}

\end{proof}

\subsection{Group structure of B and H}

\begin{theorem}
    $B$ group's structure is isomorphic to $\mathbb{Z}_2 \otimes \mathbb{Z}_2 \otimes \mathbb{Z}_2$.
    \label{theorem:Bstructure}
\end{theorem}
\begin{proof}
    Both the CZ gate and the $S$ gate belong to $\mathbb{Z}_2$ groups because they each have two possible actions: either apply the gate or do nothing (identity operation). Since these gates commute with one another and are their own inverses~(up to phase), their independent applications naturally form a direct product structure.

    For the S gate, since there are two qubits involved in bilocal operations, there are two independent $\mathbb{Z}_2$ groups corresponding to the two qubits on which the S gates can act. Specifically:
    One $\mathbb{Z}_2$ corresponds to the application of the S gate on the first qubit.
    Another $\mathbb{Z}_2$ corresponds to the application of the S gate on the second qubit.

    The CZ gate forms its own $\mathbb{Z}_2$ group because it can either be applied between the two qubits or not applied (identity action). This is independent of the application of the S gates.

    Thus, the overall group structure of the B group is isomorphic to $\mathbb{Z}_2 \otimes \mathbb{Z}_2 \otimes \mathbb{Z}_2$.
    Consequently, the order of this group is $2 \times 2 \times 2 = 8$, which accounts for all possible combinations of applying or not applying the CZ gate and the S gates across qubits in the bilocal B group configuration.
\end{proof}

\begin{theorem}
    The $H$ group is isomorphic to \( \mathbb{D}_3 \), the dihedral group of order 6.
\end{theorem}
\begin{proof}
    The $H$ group contains six elements: Identity, SWAP, and the four "CNOT family" gates. These gates correspond to the symmetries of an equilateral triangle, which include two rotations and three reflections. The $H$ group is isomorphic to \( \mathbb{D}_3 \) because \( \mathbb{D}_3 \) is the only non-Abelian group of order 6.
    A standard presentation for $\mathbb{D}_3$ is
    \begin{equation}
    \langle r, s \,\mid\, r^3 = e,\; s^2 = e,\; s\,r\,s = r^{-1} \rangle,    
    \end{equation}
        
    where $r$ can be viewed as a rotation by $120^\circ$ and $s$ as a reflection. 

    We identify $\mathit{CNOT}_{12}$ and $\mathit{CNOT}_{21}$ as two distinct reflections within $\mathbb{D}_3$. Indeed, one can verify that each $\mathit{CNOT}$ gate is its own inverse, analogous to $s^2 = e$ in the dihedral presentation. Moreover, 
    \begin{equation}
    \begin{gathered}
    \mathit{SWAP} = \mathit{CNOT}_{12}\;\mathit{CNOT}_{21}\;\mathit{CNOT}_{12} \\
    \Longleftrightarrow s_3=s_1\,s_2\,s_1
    \end{gathered}   
    \end{equation}

    mirrors the standard relation in \(\mathbb{D}_3\) that one reflection can be expressed as the composition of two others. $DCX_{12}$ and $DCX_{21}$ each plays rotation. See Table~\ref{tab:gate_transformations}
\end{proof}

\begin{table*}[!tp]
    \centering
    \renewcommand{\arraystretch}{1.5} 
    \begin{tabular}{c|cccccc}
        & \textbf{Identity} & \textbf{SWAP} & \textbf{CNOT}$_{12}$ & \textbf{DCX}$_{21}$ & \textbf{DCX}$_{12}$ & \textbf{CNOT}$_{21}$ \\
        \hline
        \textbf{Identity}   & \textbf{Identity} & \textbf{SWAP} & \textbf{CNOT}$_{12}$ & \textbf{DCX}$_{21}$ & \textbf{DCX}$_{12}$ & \textbf{CNOT}$_{21}$ \\
        \textbf{SWAP}       & \textbf{SWAP} & \textbf{Identity} & \textbf{DCX}$_{21}$ & \textbf{CNOT}$_{12}$ & \textbf{CNOT}$_{21}$ & \textbf{DCX}$_{12}$ \\
        \textbf{CNOT}$_{12}$  & \textbf{CNOT}$_{12}$ & \textbf{DCX}$_{12}$ & \textbf{Identity} & \textbf{CNOT}$_{21}$ & \textbf{SWAP} & \textbf{DCX}$_{21}$ \\
        \textbf{DCX}$_{21}$  & \textbf{DCX}$_{21}$ & \textbf{CNOT}$_{21}$ & \textbf{SWAP} & \textbf{DCX}$_{12}$ & \textbf{Identity} & \textbf{CNOT}$_{12}$ \\
        \textbf{DCX}$_{12}$  & \textbf{DCX}$_{12}$ & \textbf{CNOT}$_{12}$ & \textbf{CNOT}$_{21}$ & \textbf{Identity} & \textbf{DCX}$_{21}$ & \textbf{SWAP} \\
        \textbf{CNOT}$_{21}$  & \textbf{CNOT}$_{21}$ & \textbf{DCX}$_{21}$ & \textbf{DCX}$_{12}$ & \textbf{CNOT}$_{12}$ & \textbf{SWAP} & \textbf{Identity} \\
    \end{tabular}
    \caption{Full multiplication table of the $H$ group}
    \label{tab:gate_transformations}
\end{table*}

\par\textbf{\centering End of Proof.}

\subsection{{Extending to multi-qubit gates}}
\label{app:multiqubit}

{The argument also applies to $m$ GHZ copies shared by $n\geq3$ nodes, with an $m$-qubit Clifford gate at each node. The $Z$-type stabilizers require the same invertible binary map on the local $Z$ operators at every node. Homogeneous CNOT layers generate this map. After removing it, the remaining gates are products of $S$ and CZ, up to local Paulis. The global $X$-type stabilizers require each $S$ or CZ contribution to occur at an even number of nodes, modulo local Paulis. These contributions can therefore be grouped into paired $S$ and paired CZ layers. Each layer preserves the product GHZ basis and acts on at most two copies. Thus the existing generators suffice for more copies. Hostens et al.~\cite{hostens2006} derive the binary constraints for local Clifford operations that permute products of CSS-basis states, including GHZ states. See their Secs.~III and~VI for details.}

\section{Asymmetry of X and Z Errors in GHZ States}
\label{asymmetry}
The entanglement structure of GHZ states leads to a fundamental asymmetry between the effects of $X$ and $Z$ errors. Consider the $n$-qubit GHZ state:
\begin{equation}
\ket{\mathrm{GHZ}_n} = \frac{1}{\sqrt{2}} (\ket{0}^{\otimes n} + \ket{1}^{\otimes n}), 
\end{equation}
which is stabilized by the operators \( X^{\otimes n} \) and \( Z_i Z_{i+1} \) for all \( i \in \{1, \ldots, n-1\} \).

A single $Z$ error on any qubit flips the relative phase between the two components, mapping the state to:
\begin{equation}
Z_k \ket{\mathrm{GHZ}_n} = \frac{1}{\sqrt{2}} (\ket{0}^{\otimes n} - \ket{1}^{\otimes n}),    
\end{equation}
independently of the qubit \( k \) on which it acts. All such errors lead to the same syndrome and result in same state.

In contrast, an $X$ error on qubit \( k \) leads to:
\begin{equation}
X_k \ket{\mathrm{GHZ}_n} = \frac{1}{\sqrt{2}} (\ket{0 \ldots 1_k \ldots 0} + \ket{1 \ldots 0_k \ldots 1}),    
\end{equation}
producing a state that depends on the location of the error. This transformation leads to distinct, non-degenerate error outcomes~(See left panel of Fig.~\ref{fig:XZ}).

This asymmetry contrasts sharply with the Bell pair case, where $X$ and $Z$ errors are symmetric under local basis rotations, allowing both types of measurements to be used interchangeably in distillation protocols~(Fig.~\ref{fig:CNOT_error}). For GHZ states, however, standard protocols such as pumping and nested distillation~\cite{aschauer2005,dur2007} typically perform only $Z$-basis measurements, thereby detecting only $X$ errors. One might consider using $X$-basis measurements to detect $Z$ errors directly. However, in the presence of noise, performing a standalone $X$ basis measurement often results in lower output fidelity due to the nonlocal propagation of $Z$ errors across the entangled state~(See right panel of Fig.~\ref{fig:XZ}).

\begin{figure}[!htbp]
    \centering
    \begin{minipage}{0.46\columnwidth}
        \centering
        \begin{quantikz}
            \text{State 1} & \ctrl{1} &              \\
            \text{State 2} & \targ{}  & \meterD{Z}
        \end{quantikz}
        \\[0.5em]
        (a) Z-basis measurement.
    \end{minipage}
    \hfill
    \begin{minipage}{0.46\columnwidth}
        \centering
        \begin{quantikz}
            \text{State 1} & \targ{}   &              \\
            \text{State 2} & \ctrl{-1} & \meterD{X}
        \end{quantikz}
        \\[0.5em]
        (b) X-basis measurement.
    \end{minipage}
    \caption{\textbf{Error propagation under CNOT and basis-dependent measurements.} Depending on the measurement basis, $X$ and $Z$ errors propagate differently through the CNOT gate, leading to different detectability of errors.}
    \label{fig:CNOT_error}
\end{figure}
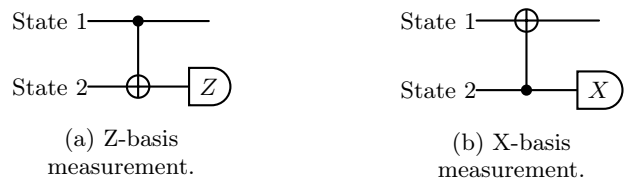

As a result, the output state after each distillation round usually exhibits a biased Pauli error, dominated by undetected $Z$ errors. To mitigate this, most protocols apply twirling after each round to symmetrize the error channel—effectively transforming the biased error into an approximate depolarizing channel. Our method avoids such artifacts by explicitly modeling the asymmetric error structure, rather than assuming it is erased via twirling.
In particular, we do not allow twirling steps at arbitrary intermediate stages of the protocol, since such probabilistic randomization is rarely implemented in physical experiments. This distinction explains why our results differ from those of nested distillation in the low-fidelity regime, while offering a more realistic assessment of experimental performance.

Formally, the Pauli twirling of a quantum channel $\mathcal{E}$ is defined as:
\begin{equation}
\mathcal{E}_\text{twirled}(\rho) = \frac{1}{|\mathcal{P}|} \sum_{P \in \mathcal{P}} P^\dagger\, \mathcal{E}(P \rho P^\dagger)\, P,
\end{equation}
where $\mathcal{P}$ denotes the $n$-qubit Pauli group. This operation maps any error channel $\mathcal{E}$ into a Pauli channel by averaging over all Pauli conjugations, thereby eliminating off-diagonal terms in the process matrix. In practice, this is done by randomly inserting Pauli gates before and after the noisy operation.

\begin{figure*}[!tp]
    \centering
    \includegraphics[width=0.9\linewidth]{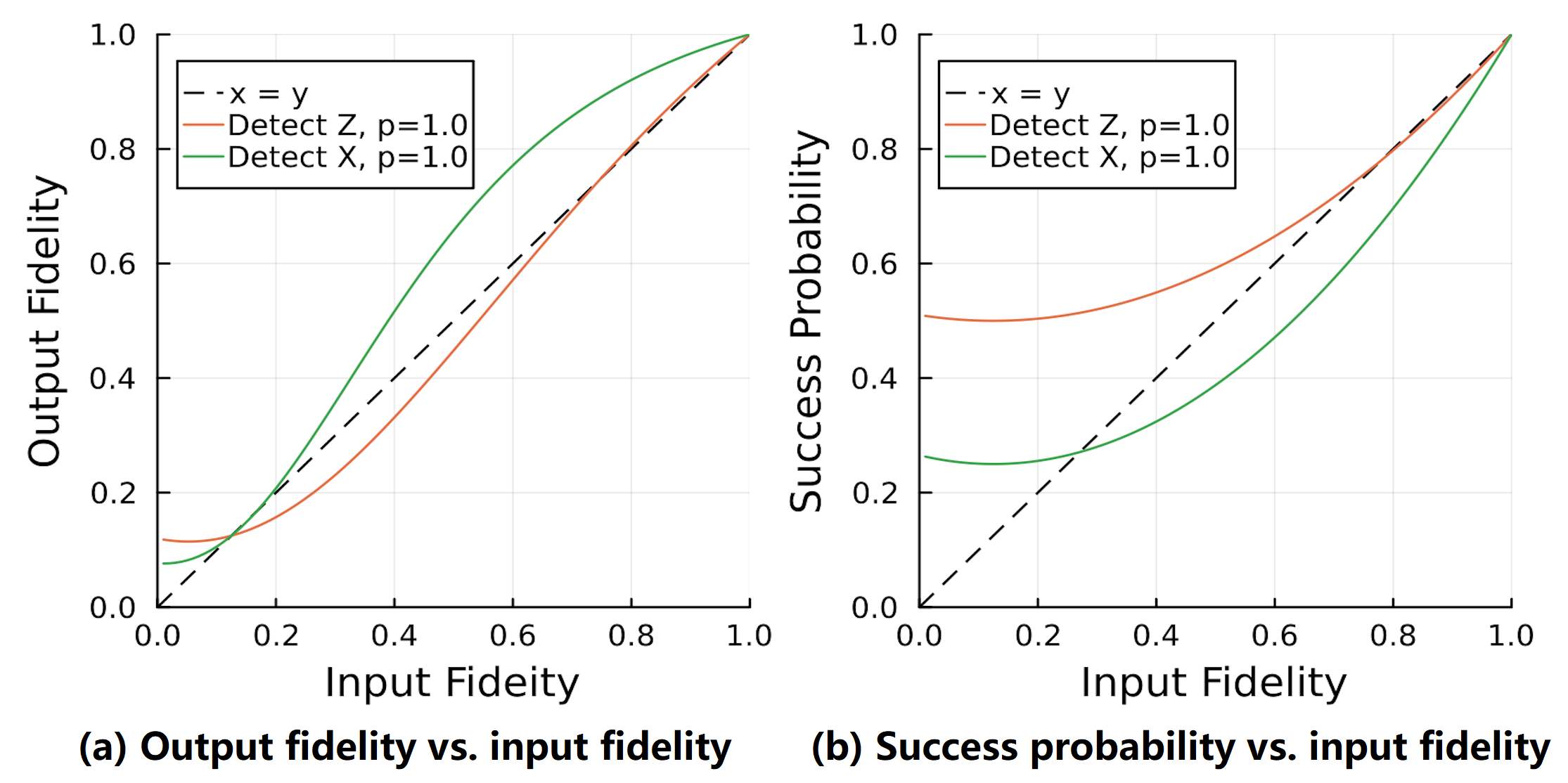}
    \caption{\textbf{Asymmetry of $X$ and $Z$ errors in a 3-qubit GHZ state.}~This figure illustrates the result of a single round of CNOT-based distillation applied to noisy GHZ states under depolarization noise. In both panels, the orange curve corresponds to attempts to detect $Z$ errors, and the green curve corresponds to attempts to detect $X$ errors. (a) shows the output fidelity as a function of input fidelity. Even under perfect gate operations~($p=0$), attempts to detect $Z$ errors~(measure in X basis) result in only minor improvement—or even degradation—of fidelity. (b) shows the success probability versus input fidelity, indicating that GHZ states with a relative phase (caused by $Z$ errors) are effectively indistinguishable during the distillation process.}
    \label{fig:XZ}
\end{figure*}

\section{Entanglement Distillation Protocols}
\label{entanglement pumping}
\subsection{Pumping Protocol}
Entanglement pumping is an established distillation protocol originally developed for bipartite entangled states such as Bell pairs~\cite{dur2007,dur2003,zoller1998,zoller1999}. It reduces the memory requirements of conventional schemes like BBPSSW or DEJMPS by avoiding the need to store multiple identical high-fidelity states simultaneously. Instead, a single stored state is repeatedly purified using freshly generated low-fidelity copies, effectively “pumping” its fidelity upward over time~(see Fig.~\ref{fig:twocolumn} for a comparison with our optimized circuits).

In each round, a raw state (e.g., one transmitted through a noisy channel) is used to distill a stored target state. If distillation succeeds, the fidelity of the target state is improved; if it fails, the process restarts from scratch using new raw states. This sequential structure greatly reduces spatial resources—only two states need to be stored at any time—but at the cost of increased temporal resources due to repeated failures and retries~(Fig.~\ref{fig:pumping}).

While originally proposed for Bell pairs, the same recurrence principle extends straightforwardly to multipartite entangled states such as GHZ states. Following the general multipartite purification framework introduced by Dür and Briegel~\cite{dur2003,dur2007}, we implement the corresponding GHZ-version of the pumping protocol as a reference baseline.
Each pumping round proceeds as follows: a stored GHZ state is paired with a newly generated raw GHZ state of input fidelity $f_{\mathrm{in}}$; for each qubit position $i$ in the two states, a CNOT gate is applied from the $i$-th qubit of stored state (control) to the raw state (target), followed by $Z$-basis measurements on all qubits of the raw state~(see Fig.~\ref{fig:CNOT_error}a). 
The measurement outcomes are used to determine whether the round is successful. Success is defined as all $Z$-basis measurement results being identical (coincidence across all measured qubits), in which case the post-measurement stored state has updated fidelity $f_{\mathrm{out}}$ (computed under the assumption of twirling, which symmetrizes $X$ and $Z$ errors); otherwise, the stored state is discarded and replaced by a fresh raw state. The process is repeated until the desired number of output GHZ states $K$ is obtained.  
All simulations use $n=3$ qubits per GHZ state, and gate error rate $p=0.01$, measurement error rate $\eta=0.01$ as in the main text.

\usetikzlibrary{fit, positioning, arrows.meta}
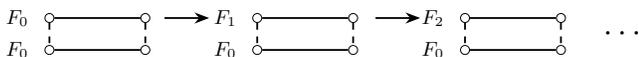
\begin{figure}[!htbp]
    \centering
    \resizebox{\linewidth}{!}{
    \begin{tikzpicture}[
    node distance=0.5cm,
    qubit/.style={circle, draw, fill=white, minimum size=4pt, inner sep=1pt},
    connection/.style={draw, thick},
    dashed_connection/.style={draw, dashed, thick},
    arrow/.style={->, thick, >=Stealth}
    ]
    \node[qubit] (l1a) at (-4, 0.5) {};
    \node[qubit] (l1b) at (-2.5, 0.5) {};
    \node[qubit] (l1c) at (-4, 0) {};
    \node[qubit] (l1d) at (-2.5, 0) {};
    \draw[connection] (l1a) -- (l1b);
    \draw[connection] (l1c) -- (l1d);
    \draw[dashed_connection] (l1a) -- (l1c);
    \draw[dashed_connection] (l1b) -- (l1d);
    \node at (-4.5, 0.5) {$F_0$};
    \node at (-4.5, 0) {$F_0$};
    \node[qubit] (m1a) at (-0.75, 0.5) {};
    \node[qubit] (m1b) at (0.75, 0.5) {};
    \node[qubit] (m1c) at (-0.75, 0) {};
    \node[qubit] (m1d) at (0.75, 0) {};
    \draw[connection] (m1a) -- (m1b);
    \draw[connection] (m1c) -- (m1d);
    \draw[dashed_connection] (m1a) -- (m1c);
    \draw[dashed_connection] (m1b) -- (m1d);
    \node at (-1.25, 0.5) {$F_1$};
    \node at (-1.25, 0) {$F_0$};
    \node[qubit] (r1a) at (2.5, 0.5) {};
    \node[qubit] (r1b) at (4, 0.5) {};
    \node[qubit] (r1c) at (2.5, 0) {};
    \node[qubit] (r1d) at (4, 0) {};
    \draw[connection] (r1a) -- (r1b);
    \draw[connection] (r1c) -- (r1d);
    \draw[dashed_connection] (r1a) -- (r1c);
    \draw[dashed_connection] (r1b) -- (r1d);
    \node at (2, 0.5) {$F_2$};
    \node at (2, 0) {$F_0$};
    \draw[arrow] (-2.2, 0.5) -- (-1.5, 0.5);
    \draw[arrow] (1.1, 0.5) -- (1.8, 0.5);
    \node at (5, 0.25) {\Large $\cdots$};
    
    \end{tikzpicture}
    }
    \caption{\textbf{Schematic illustration of entanglement pumping.}~A fresh raw state (state 1) with fixed fidelity \( F_0 \) is repeatedly generated and used to purify a stored target state (state 2). Over successive rounds, the fidelity of the target state converges to a fixed point \( F_{\mathrm{fix}} < 1 \), determined by the quality of the raw state. This process reduces memory requirements by requiring only two states to be stored at a time, at the cost of increased temporal overhead.}
    \label{fig:pumping}
\end{figure}

\subsection{Nested (Recursive) Protocol}

Nested distillation is a recursive extension of entanglement pumping that aims to achieve higher output fidelities by organizing multiple purification rounds into a hierarchical structure~\cite{dur2003}. Instead of purifying a single target state using a stream of fresh states, nested protocols use multiple raw states to produce several intermediate purified states, which are then distilled again in subsequent rounds~(see Fig.~\ref{fig:nest} for a comparison with our optimized circuits).

For instance, a two-level nested protocol might begin by grouping $N$ noisy states into $N/2$ pairs, each undergoing a purification step as pumping. For each qubit position $i$ in the two states, a CNOT gate is applied from the $i$-th qubit of one state (control) to the $i$-th qubit of the other state (target), followed by $Z$-basis measurements on all qubits of the target state~(Fig.~\ref{fig:CNOT_error}a).
A round is considered successful only if all $Z$-basis measurement results are identical (coincidence across all measured qubits), in which case the remaining state’s fidelity $f_{\mathrm{out}}$ is computed under the assumption of twirling, which symmetrizes $X$ and $Z$ errors.  
If unsuccessful, both states are discarded. 
The surviving intermediate states are then grouped again into pairs and distilled in the same manner, until a single high-fidelity state is produced.  

The number of required raw states grows exponentially with the number of nesting levels, making these protocols memory-intensive and less flexible in practice (Fig.~\ref{fig:nested}).  
Despite their theoretical effectiveness, nested protocols often assume idealized conditions, including perfect Clifford gates, ideal measurements, and symmetric error channels.  
Moreover, most implementations rely on $Z$-basis measurements only, detecting primarily $X$ errors while leaving $Z$ errors unaddressed.  
This leads to biased intermediate states that require twirling to restore symmetric error profiles, as discussed in Appendix~\ref{asymmetry}.

In our simulations, we compare our optimized circuits against nested protocols with two and three levels of purification, using $4$ and $8$ raw GHZ states respectively. We set $n=3$ qubits per GHZ state, gate error rate $p=0.01$, and measurement error rate $\eta=0.01$ as in the main text. The results highlight the trade-off between fidelity, success probability, and resource consumption across different distillation strategies.

For completeness, Fig.~\ref{fig:sequence} compares the basic two colorable graph recurrence protocol with the nested version~\cite{aschauer2005}. The figure visualizes three consecutive purification rounds with different measurement-basis combinations ($ZZZ$-$XXX$).
In general two-colorable graph states, the optimal purification order may involve non-alternating sequences of P1 and P2.
However, for GHZ states—which are highly asymmetric two-colorable graphs with only one qubit in one color set—the combination of twirling and repeated $Z$-basis measurements yields better performance. This confirms that using the nested protocol as the reference baseline in our main-text comparisons provides a fair and representative benchmark for evaluating GHZ-preserving circuits.

\begin{figure}[!htbp]
    \centering
    \resizebox{\linewidth}{!}{
    \begin{tikzpicture}[
    node distance=0.5cm,
    qubit/.style={circle, draw, fill=white, minimum size=4pt, inner sep=1pt},
    connection/.style={draw, thick},
    dashed_connection/.style={draw, dashed, thick},
    arrow/.style={->, thick, >=Stealth}
    ]
    
    \node at (-6, 3) {\Large $F_0$};
    
    \node[qubit] (l1a) at (-7, 2) {};
    \node[qubit] (l1b) at (-5, 2) {};
    \node[qubit] (l1c) at (-7, 1.5) {};
    \node[qubit] (l1d) at (-5, 1.5) {};
    \draw[connection] (l1a) -- (l1b);
    \draw[connection] (l1c) -- (l1d);
    \draw[dashed_connection] (l1a) -- (l1c);
    \draw[dashed_connection] (l1b) -- (l1d);
    
    \node[qubit] (l2a) at (-7, 0.5) {};
    \node[qubit] (l2b) at (-5, 0.5) {};
    \node[qubit] (l2c) at (-7, 0) {};
    \node[qubit] (l2d) at (-5, 0) {};
    \draw[connection] (l2a) -- (l2b);
    \draw[connection] (l2c) -- (l2d);
    \draw[dashed_connection] (l2a) -- (l2c);
    \draw[dashed_connection] (l2b) -- (l2d);
    
    \node[qubit] (l3a) at (-7, -1) {};
    \node[qubit] (l3b) at (-5, -1) {};
    \node[qubit] (l3c) at (-7, -1.5) {};
    \node[qubit] (l3d) at (-5, -1.5) {};
    \draw[connection] (l3a) -- (l3b);
    \draw[connection] (l3c) -- (l3d);
    \draw[dashed_connection] (l3a) -- (l3c);
    \draw[dashed_connection] (l3b) -- (l3d);
    
    \node[qubit] (l4a) at (-7, -2.5) {};
    \node[qubit] (l4b) at (-5, -2.5) {};
    \node[qubit] (l4c) at (-7, -3) {};
    \node[qubit] (l4d) at (-5, -3) {};
    \draw[connection] (l4a) -- (l4b);
    \draw[connection] (l4c) -- (l4d);
    \draw[dashed_connection] (l4a) -- (l4c);
    \draw[dashed_connection] (l4b) -- (l4d);
    
    \node at (-2, 3) {\Large $F_1$};
    
    \node[qubit] (m1a) at (-3, 1.25) {};
    \node[qubit] (m1b) at (-1, 1.25) {};
    \node[qubit] (m1c) at (-3, 0.75) {};
    \node[qubit] (m1d) at (-1, 0.75) {};
    \draw[connection] (m1a) -- (m1b);
    \draw[connection] (m1c) -- (m1d);
    \draw[dashed_connection] (m1a) -- (m1c);
    \draw[dashed_connection] (m1b) -- (m1d);
    
    \node[qubit] (m2a) at (-3, -1.75) {};
    \node[qubit] (m2b) at (-1, -1.75) {};
    \node[qubit] (m2c) at (-3, -2.25) {};
    \node[qubit] (m2d) at (-1, -2.25) {};
    \draw[connection] (m2a) -- (m2b);
    \draw[connection] (m2c) -- (m2d);
    \draw[dashed_connection] (m2a) -- (m2c);
    \draw[dashed_connection] (m2b) -- (m2d);
    
    \node at (2, 3) {\Large $F_2$};
    
    \node[qubit] (r1a) at (1, -0.25) {};
    \node[qubit] (r1b) at (3, -0.25) {};
    \node[qubit] (r1c) at (1, -0.75) {};
    \node[qubit] (r1d) at (3, -0.75) {};
    \draw[connection] (r1a) -- (r1b);
    \draw[connection] (r1c) -- (r1d);
    \draw[dashed_connection] (r1a) -- (r1c);
    \draw[dashed_connection] (r1b) -- (r1d);
    
    \draw[arrow] (-4.5, 1.75) -- (-3.5, 1.75);
    \draw[arrow] (-4.5, 0.25) -- (-3.5, 0.25);
    \draw[arrow] (-4.5, -1.25) -- (-3.5, -1.25);
    \draw[arrow] (-4.5, -2.75) -- (-3.5, -2.75);
    
    \draw[arrow] (-0.5, 1.0) -- (0.5, 1.0);
    \draw[arrow] (-0.5, -2.0) -- (0.5, -2.0);
    
    \end{tikzpicture}
    }
    \caption{\textbf{Schematic illustration of nested distillation.}~Multiple noisy states are first grouped and purified in parallel at the first level, producing intermediate states of higher fidelity. These intermediate states are then recursively grouped and distilled in higher-level rounds, forming a nested hierarchy. Each level improves fidelity at the cost of exponentially increasing resource consumption: a two-level nesting requires at least four input states, and a three-level nesting requires at least eight. Compared to pumping, nested protocols achieve higher final fidelities but at the cost of significantly increased register usage.}
    \label{fig:nested}
\end{figure}
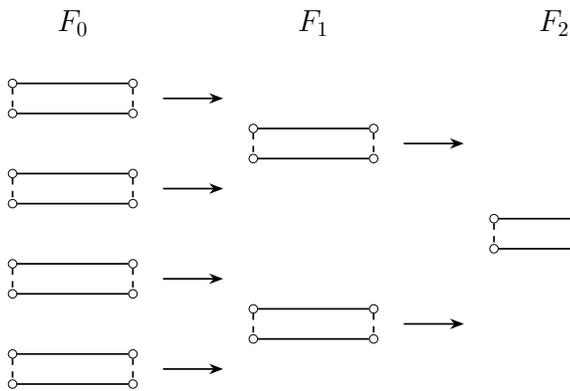

\begin{figure}[!htbp]
    \centering
    \includegraphics[width=0.9\linewidth]{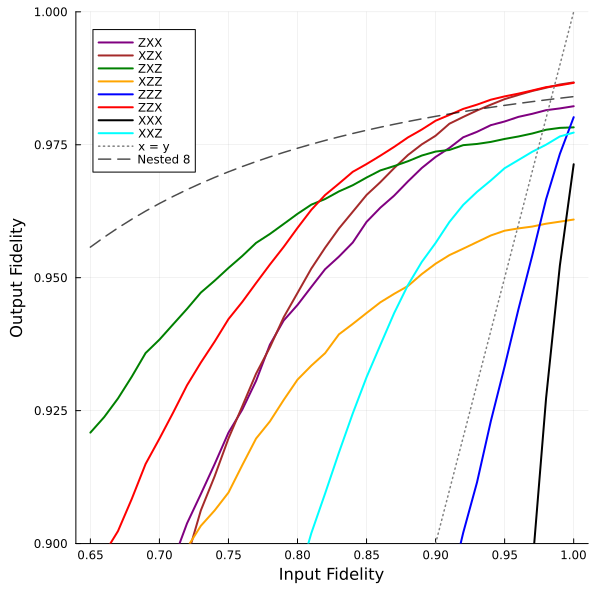}
    \caption{\textbf{Comparison between the standard two-colorable-graph nested (recurrence) protocols (P1–P2) and the nested protocol we used as a reference in the main text.}~The P1–P2 scheme, originally introduced for two-colorable graph states~\cite{aschauer2005}, alternates local CNOT operations between vertex sets $V_A$ and $V_B$ and performs postselection on the measurement outcomes of one copy. The additional nested protocol we present inserts twirling between successive rounds, which symmetrizes the accumulated errors. The nested protocol with twirling achieves higher output fidelities and success probabilities than the standalone P1–P2 scheme.
    The legend labels ($ZZZ$, $ZZX$, $ZXZ$, $ZXX$, $XZZ$, $XZX$, $XXZ$, $XXX$) correspond to the measurement bases applied in each of the three successive distillation rounds, (equivalently indicating the sequence of P1 and P2 subprotocols applied in each run). In the asymptotic case the P1-P2 recurrence protocol will outperform the "nested with twirling" protocol, but for the particular case of a GHZ state and only 3 rounds of purification, "nested with twirling" works better -- this is chiefly due to the very "biased" coloring of a GHZ graph state (one single qubit of color A, while all other qubits are of color B).
    All data in this figure assume a local gate error rate $p=0.01$ and a measurement error rate $\eta=0.01$, matching the noise parameters used in the main text. We therefore adopt the nested protocol with twirling, rather than the elementary P1–P2 recurrence, as the baseline for a fair comparison with our optimized GHZ-preserving distillation circuits.}
    \label{fig:sequence}
\end{figure}

\subsection{Hashing and Breeding Protocol}
For completeness, we briefly discuss the asymptotic entanglement purification protocols known as hashing and breeding schemes. Both protocols operate in the limit of infinitely many copies $N \rightarrow \infty$, where (in the absence of gate noise) they can get arbitrarily close to unit fidelity of the output states~\cite{bennett1996mixed,aschauer2005,bennett1996,dur2007}.

The breeding protocol assumes that the participating parties already share a number of perfectly entangled states in addition to the noisy copies to be purified. These auxiliary pure states are used to extract classical information about the error configuration of the ensemble without destroying any of the noisy states. Once all error syndromes are learned, appropriate local corrections yield a collection of purified states. While conceptually elegant, this approach relies on the unrealistic assumption that a sufficient supply of pre-purified GHZ (or Bell) states is available at the outset~\cite{bennett1996}.

The hashing protocol removes the need for pre-purified ancillas but requires simultaneous access to a large ensemble of noisy states. By sacrificing a subset of copies through joint local operations and measurements, one learns parity information about the remaining states and can deterministically correct them in the asymptotic limit. The theoretical yield of both hashing and breeding is $D=1-S(\rho)$, where $S(\rho)$ is the von Neumann entropy of the initial mixed state~\cite{bennett1996mixed,Maneva2000}.

In practice, neither protocol is well suited for hardware-limited architectures and realistic noisy operations. Gate errors and measurement imperfections rapidly degrade the accuracy of the parity information that these asymptotic schemes rely on. Moreover, breeding presupposes ideal ancillary resources that are unavailable in realistic quantum networks, while hashing demands a large number of qubit registers to store and manipulate many GHZ states simultaneously. Since the methods developed in this work focus on finite-size, circuit-level GHZ distillation under realistic noise and register constraints, these asymptotic schemes are included here only for theoretical comparison and are not used as performance benchmarks.

{
\section{Circuit Optimization and Fitness Evaluation}
\label{app:cost_function}

A candidate circuit is an ordered sequence
\begin{equation}
\mathcal{C}=(o_1,\ldots,o_L),
\qquad L\leq L_{\max},
\label{eq:candidate_circuit}
\end{equation}
whose operations are H-group gates, B-group gates, or $X$- and $Z$-basis measurements on non-output registers. The first $K$ registers are retained as outputs. An $X$-basis measurement is accepted when its global-$X$ phase bit is zero, while a $Z$-basis measurement is accepted when all of its $ZZ$-syndrome bits are zero; any other result rejects that trial. After an accepted intermediate measurement, the measured register can be loaded with a new raw GHZ state, allowing $N>R$ through register reuse. The operation sequence, measurement bases, and postselection conditions are fixed before execution, and adaptive feed-forward is not included.

The circuit representation and cost-function evaluation do not depend on the optimization heuristic. We used a genetic algorithm with generation-dependent fitness evaluation for the reported numerical examples. Initial candidates contained one to four randomly selected H- or B-group gates followed by $X$- or $Z$-basis measurements on the non-output registers. The retained population contained 30 candidates and was evolved for 20 generations, with $L_{\max}=20$. Each generation used 50 randomly selected parent pairs and two offspring per pair. Insertion, deletion, reordering, and operation-mutation moves were each attempted with probability $0.2$. Candidate fitness in generation $g$ was estimated using
\begin{equation}
T_g=1000g,
\qquad g=1,\ldots,20,
\label{eq:generation_trials}
\end{equation}
Monte Carlo trials. The increasing sampling budget reduces the cost of early exploration and gives more precise fitness estimates in later generations. The generation dependence enters through this evaluation schedule; the search does not use a Metropolis acceptance rule. No random seed was fixed, so independent searches need not return identical circuits. The implementation, selected circuits, and numerical outputs used here are available in the public archive~\cite{githubghz}.

For a candidate evaluated with $T_g$ trials, let $N_{\mathrm{acc}}(\mathcal{C})$ be the number that pass every postselection condition and let $N_{\mathrm{tar}}(\mathcal{C})$ be the number of accepted trials whose retained register is in the target GHZ state. The estimated success probability and conditional output fidelity are
\begin{equation}
\widehat{P}_{\mathrm{succ}}(\mathcal{C})
=\frac{N_{\mathrm{acc}}(\mathcal{C})}{T_g},
\qquad
\widehat{f}_{\mathrm{out}}(\mathcal{C})
=\frac{N_{\mathrm{tar}}(\mathcal{C})}{N_{\mathrm{acc}}(\mathcal{C})}.
\label{eq:fitness_estimators}
\end{equation}
In an unconstrained search, candidates are ranked by $\widehat{f}_{\mathrm{out}}$. For comparisons imposing a minimum success probability, only candidates satisfying
\begin{equation}
\widehat{P}_{\mathrm{succ}}(\mathcal{C})\geq P_{\min}
\label{eq:success_filter}
\end{equation}
are retained for comparison. Here $P_{\min}$ is the success probability required by the matched baseline.
}

{
\section{Affine Binary Action of GHZ-Preserving Gates}
\label{app:binary_action}

This appendix gives a direct binary representation of the H- and B-group actions. It makes explicit why a general stabilizer tableau is not required once a state is known to remain in the product GHZ basis.

For GHZ copy $r\in\{1,2\}$, define the phase label
\begin{equation}
\boldsymbol{\lambda}^{(r)}
=
\bigl(x_r,z_{r,1},\ldots,z_{r,n-1}\bigr)^{\mathsf T}
\in \mathbb{F}_2^n,
\label{eq:single_ghz_label}
\end{equation}
where $x_r$ records the eigenvalue $(-1)^{x_r}$ of $X^{\otimes n}$ and $z_{r,j}$ records the eigenvalue $(-1)^{z_{r,j}}$ of $Z_jZ_{j+1}$. The joint label of two copies is
\begin{equation}
\boldsymbol{\lambda}
=
\begin{pmatrix}
\boldsymbol{\lambda}^{(1)}\\
\boldsymbol{\lambda}^{(2)}
\end{pmatrix}
\in \mathbb{F}_2^{2n}.
\end{equation}
Because Clifford conjugation maps each stabilizer generator to a product of stabilizer generators, multiplication of eigenvalue signs becomes addition modulo two. Every GHZ-preserving Clifford operation therefore induces an affine binary map
\begin{equation}
\boldsymbol{\lambda}'
=M_g\boldsymbol{\lambda}\oplus\boldsymbol{c}_g,
\qquad
M_g\in \operatorname{GL}(2n,\mathbb{F}_2),
\quad
\boldsymbol{c}_g\in\mathbb{F}_2^{2n}.
\label{eq:affine_binary_action}
\end{equation}
The offset $\boldsymbol{c}_g$ records signs introduced by the chosen physical Clifford representative. Explicit Pauli operations also change only this offset. A different representative of the same phaseless Clifford coset can therefore change $\boldsymbol{c}_g$ without changing $M_g$. Composition is also affine:
\begin{equation}
(M_2,\boldsymbol{c}_2)\circ(M_1,\boldsymbol{c}_1)
=
\bigl(M_2M_1,M_2\boldsymbol{c}_1\oplus\boldsymbol{c}_2\bigr).
\end{equation}

\subsection{H-group action}

Let
\begin{equation}
\boldsymbol{x}=\begin{pmatrix}x_1\\x_2\end{pmatrix},
\qquad
\boldsymbol{z}_j=\begin{pmatrix}z_{1,j}\\z_{2,j}\end{pmatrix}.
\end{equation}
Each homogeneous gate $h\in H$ corresponds to a matrix $A_h\in\operatorname{GL}(2,\mathbb{F}_2)$ and acts as
\begin{equation}
\boldsymbol{x}'=A_h\boldsymbol{x},
\qquad
\boldsymbol{z}_j'=(A_h^{-1})^{\mathsf T}\boldsymbol{z}_j
\quad (j=1,\ldots,n-1).
\label{eq:H_binary_action}
\end{equation}
The six elements of $\operatorname{GL}(2,\mathbb{F}_2)$ reproduce the six H-group elements. For example, a transversal $\mathit{CNOT}_{12}$ has
\begin{equation}
A_{12}=\begin{pmatrix}1&1\\0&1\end{pmatrix},
\end{equation}
and hence
\begin{align}
x_1'&=x_1\oplus x_2, & x_2'&=x_2,\notag\\
z_{1,j}'&=z_{1,j}, & z_{2,j}'&=z_{1,j}\oplus z_{2,j}.
\label{eq:CNOT12_label_action}
\end{align}
Similarly, $\mathit{CNOT}_{21}$ reverses the roles of the two copies, and SWAP exchanges their complete phase labels. The two DCX elements are obtained by composing the two CNOT updates. Thus every H-group update uses only swaps and XORs.

\subsection{B-group action}

Consider a B-group element applied at nodes $j$ and $j+1$. Let $a,b,c\in\mathbb{F}_2$ indicate, respectively, paired phase gates on copy 1, paired phase gates on copy 2, and paired CZ gates between the two copies. For the standard physical convention $S=\operatorname{diag}(1,i)$, the $Z$-type phase bits are unchanged, while the global $X$-type bits transform as
\begin{align}
x_1' &= x_1\oplus a z_{1,j}\oplus c z_{2,j}\oplus a,\notag\\
x_2' &= x_2\oplus b z_{2,j}\oplus c z_{1,j}\oplus b,\notag\\
z_{r,k}' &= z_{r,k}
\quad (r\in\{1,2\},\;k=1,\ldots,n-1).
\label{eq:B_binary_action}
\end{align}
The eight choices of $(a,b,c)$ give the eight B-group elements. The final $a$ and $b$ terms in the first two lines are affine offsets produced by the standard $S$-gate convention. They can instead be absorbed into the Pauli layer when a different representative of the same phaseless coset is used. Equation~\eqref{eq:B_binary_action} also shows directly how a B-group operation transfers a selected $ZZ$-phase bit into the phase of the global $X^{\otimes n}$ stabilizer.

\subsection{Measurement maps and postselection}

Measurement reads the phase bits of the corresponding GHZ stabilizer generators. For copy $r$, define the projections from the joint label by
\begin{equation}
\pi_X^{(r)}\!\left(\boldsymbol{\lambda}\right)=x_r,
\qquad
\pi_Z^{(r)}\!\left(\boldsymbol{\lambda}\right)
=\left(z_{r,1},\ldots,z_{r,n-1}\right)^{\mathsf T}.
\label{eq:measurement_maps}
\end{equation}
If each local outcome is encoded as a bit $m_i^{(b)}$ with eigenvalue $(-1)^{m_i^{(b)}}$, an $X$-basis measurement gives
$x_r=\bigoplus_{i=1}^{n}m_i^{(X)}$, while a $Z$-basis measurement gives
$z_{r,j}=m_j^{(Z)}\oplus m_{j+1}^{(Z)}$. The accepted syndromes are therefore
$\pi_X^{(r)}(\boldsymbol{\lambda})=0$ for an $X$-basis measurement and
$\pi_Z^{(r)}(\boldsymbol{\lambda})=\boldsymbol{0}$ for a $Z$-basis measurement.

After a GHZ-preserving gate $g$, the measured syndrome is
\begin{equation}
\boldsymbol{s}_b^{(r)}
=\pi_b^{(r)}\!\left(M_g\boldsymbol{\lambda}\oplus\boldsymbol{c}_g\right),
\qquad b\in\{X,Z\},
\label{eq:propagated_measurement_map}
\end{equation}
An input-label difference $\delta\boldsymbol{\lambda}$ is undetected when
$\pi_b^{(r)} M_g\delta\boldsymbol{\lambda}=\boldsymbol{0}$. Thus, the undetected errors form the kernel of the propagated measurement map. The circuit search selects gates that move error components correlated with errors on the retained copy into measured phase bits, so postselection removes the corresponding nonzero-syndrome events from the conditional output distribution.

\subsection{Split lookup and direct packed evaluation}

Writing the linear part of Eq.~\eqref{eq:affine_binary_action} in block form gives
\begin{equation}
M_g\boldsymbol{\lambda}
=M_{g,1}\boldsymbol{\lambda}^{(1)}
\oplus M_{g,2}\boldsymbol{\lambda}^{(2)}.
\end{equation}
If a lookup-table implementation is desired, the joint table of $2^{2n}$ inputs can therefore be replaced by two contribution tables,
\begin{equation}
T_{g,1}[u]=M_{g,1}u,
\qquad
T_{g,2}[v]=M_{g,2}v,
\end{equation}
each containing $2^n$ entries. When a complete $2n$-bit output contribution fits in one 64-bit word, their combined memory is $2\times 8\times 2^n=16\,2^n$ bytes. The output is reconstructed as
\begin{equation}
\boldsymbol{\lambda}'
=T_{g,1}[\boldsymbol{\lambda}^{(1)}]
\oplus T_{g,2}[\boldsymbol{\lambda}^{(2)}]
\oplus\boldsymbol{c}_g.
\label{eq:split_lookup}
\end{equation}
More directly, Eqs.~\eqref{eq:H_binary_action} and~\eqref{eq:B_binary_action} can be evaluated on packed machine words without constructing either lookup table. This is the implementation used in the complexity analysis below.

\section{Computational Complexity and Memory Scaling}
\label{app:complexity}

We distinguish the memory and preprocessing cost of a dense permutation representation from the online cost of applying an already represented gate.

\subsection{Dense permutation table}

A gate acting on two $n$-bit GHZ phase labels permutes $2^{2n}$ possible joint labels. If each output label is stored as a 64-bit integer, one dense permutation requires
\begin{equation}
M_{\mathrm{dense}}(n)=8\,2^{2n}\ \text{bytes},
\label{eq:dense_memory}
\end{equation}
excluding array metadata, alignment, and any additional cached gates. Table~\ref{tab:dense_memory} gives representative memory requirements for a single permutation.

\begin{table}[H]

\centering
\begin{tabular}{c r r}
\toprule
$n$ & Number of entries & Memory(64-bit entries) \\
\midrule
5  & $2^{10}$ & 8 KiB \\
10 & $2^{20}$ & 8 MiB \\
15 & $2^{30}$ & 8 GiB \\
18 & $2^{36}$ & 512 GiB \\
19 & $2^{38}$ & 2 TiB \\
20 & $2^{40}$ & 8 TiB \\
\bottomrule
\end{tabular}
\caption{{Memory required for one dense $2^{2n}$-entry permutation with 64-bit integer entries. The values exclude container overhead and additional cached permutations.}}
\label{tab:dense_memory}
\end{table}

Constructing a dense permutation visits all $2^{2n}$ labels. Using the affine rules in Appendix~\ref{app:binary_action}, the worst-case total preprocessing time is
\begin{equation}
\begin{aligned}
T_{\mathrm{pre}}^{\mathrm{unpacked}}(n)
&=\Theta\!\left(n2^{2n}\right), \\
T_{\mathrm{pre}}^{\mathrm{packed}}(n,w)
&=\Theta\!\left(2^{2n}\left\lceil\frac{n}{w}\right\rceil\right).
\end{aligned}
\label{eq:dense_preprocessing}
\end{equation}
When each copy's phase label fits in one machine word, the packed preprocessing time is $\Theta(2^{2n})$.

\subsection{Packed affine evaluation}

The dense table is optional. Equations~\eqref{eq:H_binary_action} and~\eqref{eq:B_binary_action} act directly on the phase labels by swaps, masks, shifts, and XORs. With an unpacked Boolean vector, an H-group update touches $O(n)$ phase bits, so the worst-case gate-application cost is
\begin{equation}
T_{\mathrm{unpacked}}(n)=O(n).
\end{equation}

On a word-RAM with word size $w$, an $n$-bit GHZ label occupies $\lceil n/w\rceil$ words. Applying an H-group element requires a constant number of word-wise swaps and XORs per stored word, while a B-group element changes only selected fields. The worst-case cost is therefore
\begin{equation}
T_{\mathrm{packed}}(n,w)
=O\!\left(\left\lceil\frac{n}{w}\right\rceil\right)
=O(n/w).
\label{eq:packed_complexity}
\end{equation}
When each copy's phase label fits in one machine word, the update is $O(1)$ in the word-RAM model. The packed representation uses $O(\lceil n/w\rceil)$ words per GHZ copy.

}

\end{document}